\documentclass[english, 12pt]{article}
\usepackage{amsmath}
\usepackage{todonotes}
\usepackage{comment}
\usepackage{xspace}

\newcommand{\ord}{\text{ord}}

\usepackage{algorithm}% http://ctan.org/pkg/algorithms
\usepackage{algpseudocode}% http://ctan.org/pkg/algorithmicx\usepackage{algpseudocode}                 % \Require, \Ensure, \State

\usepackage{base}
\usepackage{macros}
\usepackage{minimacros}

\title{Polynomial Identity Testing and Reconstruction for Depth-4 Powering Circuits of High Degree}
\author{Amir Shpilka\thanks{School of Computer Science and AI, Tel Aviv University. This research was funded by the European Union (ERC, EACTP, 101142020). Views and opinions expressed are however those of the author(s) only and do not necessarily reflect those of the European Union or the European Research Council Executive Agency. Neither the European Union nor the granting authority can be held responsible for them.}
\and Yann Tal\footnotemark[1] }
\date{}

\begin{document}

\maketitle

\begin{abstract}
    We study deterministic polynomial identity testing (PIT) and reconstruction algorithms for depth-$4$ arithmetic circuits of the form
\[
\Sigma^{[r]}\!\wedge^{[d]}\!\Sigma^{[s]}\!\Pi^{[\delta]}.
\]
This model generalizes Waring decompositions and diagonal circuits, and captures sums of powers of low-degree sparse polynomials.  Specifically, each circuit computes a sum of $r$ terms, where each term is a $d$-th power of an $s$-sparse polynomial of degree $\delta$.
This model also includes algebraic representations that arise in tensor decomposition and moment-based learning tasks such as mixture models and subspace learning.

We give deterministic worst-case algorithms for PIT and reconstruction in this model. Our PIT construction applies when $d>r^2$ and yields explicit hitting sets of size  $O(r^4 s^4 n^2  d \delta^3)$. The reconstruction algorithm runs in time $\poly(n,s,d)$ under the condition $d=\Omega(r^4\delta)$, and in particular it tolerates polynomially large top fan-in~$r$ and bottom degree~$\delta$.  

Both results hold over fields of characteristic zero and over fields of sufficiently large characteristic. These algorithms provide the first polynomial-time deterministic solutions for depth-$4$ powering circuits with unbounded top fan-in.  In particular, the reconstruction result improves upon previous work which required non-degeneracy or average-case assumptions~\cite{GKS20,bafna2022polynomial,ChandraGKMS24}.

The PIT construction relies on the \emph{ABC theorem for function fields} (Mason-Stothers theorem), which ensures linear independence of high-degree  powers of sparse polynomials after a suitable projection.  
The reconstruction algorithm combines this with \emph{Wronskian-based differential operators}, structural properties of their kernels,  and a robust version of the Klivans-Spielman hitting set.

\end{abstract}

\thispagestyle{empty}

\newpage

\tableofcontents
\thispagestyle{empty}
\newpage

\clearpage
\pagenumbering{arabic}
\setcounter{page}{1}

\section{Introduction}

Arithmetic circuits are directed acyclic graphs that compute multivariate polynomials using addition and multiplication gates. 
They form the standard model for studying algebraic computation, serving as the analogue of Boolean circuits in the arithmetic world.

The \emph{Polynomial Identity Testing} (PIT) problem asks, given an arithmetic circuit that computes a multivariate polynomial $f(\vx)$, to decide whether $f \equiv 0$.  
Randomized algorithms for PIT have been known for decades, via the polynomial identity lemma,\footnote{Often referred to as the Schwartz--Zippel--DeMillo--Lipton--Ore lemma~\cite{Ore1922,DemilloL78,Zippel79,Schwartz80}.} yet obtaining a deterministic polynomial-time algorithm remains a fundamental open problem in derandomization.  
Hardness-randomness tradeoff results have shown a tight connection between PIT and lower bounds for arithmetic circuits~\cite{HS80,KabanetsI04,DvirSY09,Chou0S19,GuoKSS22,KumarST23bootstrap,kumar2019hardness}.  
PIT has received a great deal of attention in recent years.
Deterministic algorithms are now known for several restricted models, 
such as bounded-depth circuits, bounded-read formulas, and read-once algebraic branching programs.
Variants have also been studied for orbits of these classes under affine transformations.
For surveys on PIT, see~\cite{SY10,saxena2009progress,saxena2014progress,DuttaGhosh-survey}.

Closely related to PIT is the \emph{reconstruction problem}.  
Instead of merely determining whether $f$ is identically zero, reconstruction seeks to recover, using only black-box access to $f$, a circuit from a prescribed class that computes it.  
A deterministic reconstruction algorithm immediately implies deterministic PIT for that class, while a black-box PIT algorithm guarantees an \emph{information-theoretic} form of reconstruction.  
If $H$ is a hitting set for circuits of size $2s$, then the evaluation map
\[
\mathrm{Eval}_H : C \longmapsto (C(a))_{a\in H}
\]
is injective for all circuits $C$ of size at most $s$.  
The remaining challenge is to perform this inversion efficiently.  
This step can be significantly more difficult: there are natural circuit classes for which deterministic PIT is known, yet no efficient reconstruction algorithm is available.  
There have been many works studying the reconstruction problem in different models~\cite{beimel2000learning,klivans2006learning,ShpilkaV14,ForbesS12,GuptaKL12,GuptaKQ14}, with a recent flurry of algorithms for small-depth circuits~\cite{Shpilka09-rec-sps,KarninS09rec,Sinha16,BhargavaSV21,Sinha22,peleg2022tensor,BhargavaGKS22,BhargavaS25,BhargavaSV25,SarafSV25}.

A series of depth-reduction results~\cite{AgrawalV08,Koiran12,Tavenas15,GuptaKKS16chasm}  
proved that any polynomial computed by a small circuit can also be represented by a depth-$4$ circuit of subexponential size (and, over characteristic zero, by a depth-$3$ circuit).  
Thus, providing PIT or reconstruction algorithms for shallow arithmetic circuits is an outstanding task that can have implications for general models of arithmetic circuits.  

The search for efficient reconstruction is bounded by fundamental hardness results.  
Starting with the work of H{\aa}stad~\cite{Hastad90}, it was shown that even computing tensor rank, i.e., determining the smallest top fan-in of a set-multilinear depth-$3$ circuit is NP-hard over~$\mathbb{Q}$ and NP-complete over finite fields, even for degree-$3$ tensors.  
Extensions of this result to approximation algorithms and to other models of computation were given in subsequent papers~\cite{SchaeferStefankovic18,FortnowKlivans09,KlivansSherstov09,swernofsky2018tensor}.  

Because the tensor-rank problem is equivalent to minimizing the top fan-in of a set-multilinear $\Sigma\Pi\Sigma$ circuit, these hardness results extend directly to reconstruction: proper learning for such circuits is NP-hard and, in some instances, even undecidable \cite{ChillaraGS23}.  
This motivated the search for efficient reconstruction algorithms for certain parameter regimes or under additional restrictions on the circuit.

Reconstruction also connects naturally to \emph{learning theory}.  
An arithmetic circuit can be viewed as an algebraic hypothesis, and black-box evaluations of $f$ correspond to examples (or membership queries). 
Thus, reconstruction corresponds to \emph{exact learning} in the noiseless setting, i.e., when we obtain the exact value of $f$ on each and every query. This stands in contrast to the \emph{approximate} or \emph{noisy} learning problems that are typical in machine learning.  
This analogy has recently become concrete through work showing that some important learning tasks, such as learning mixtures of Gaussians or subspace clustering, reduce to learning polynomials of the form
\begin{equation}\label{eq:swsp}
    f(\vx) = \alpha_1 f_1(\vx)^d + \cdots + \alpha_r f_r(\vx)^d,
\end{equation}
where the $f_i$ are polynomials of degree $\delta$ \cite{GeHK15,GKS20}.  
A representation as in \eqref{eq:swsp} is also called a $\Sigma^{[r]}\!\wedge^{[d]}\!\Sigma\! \Pi^{[\delta]}$ circuit.  
When $\delta=1$, the reconstruction problem is equivalent to symmetric tensor decomposition (Waring rank decomposition), which is fundamental in many applications in machine learning, e.g., in moment-based methods for latent variable models, cf.\ \cite{goyal2014fourier,anandkumar2014tensor,montanari2014statistical,hopkins2015tensor,ma2016polynomial,hopkins2016fast,kothari2025smooth}. When $\delta=2$, the problem corresponds to learning mixtures of Gaussians and has been extensively studied, cf.\ \cite{sanjeev2001learning, dasgupta2007probabilistic,anderson2014more, regev2017learning, liu2021settling, di2023multidimensional}.

\sloppy In recent years, an increasing number of works consider the case $\delta>  2$  \cite{GKS20,ChandraGKMS24,bafna2022polynomial,bhaskara2024new}.
Garg, Kayal, and Saha~\cite{GKS20}, and subsequent work by Bafna, Hsieh, Kothari, and Xu \cite{bafna2022polynomial}, gave randomized reconstruction algorithms for such polynomials in the noiseless case, running in time $\poly(n,d,r)^\delta$.  
Their algorithm works under a non-degeneracy assumption that holds with high probability for random polynomials but does not apply in the worst case.  
Chandra, Garg, Kayal, Mittal, and Sinha~\cite{ChandraGKMS24} extended this connection to noisy and smoothed settings, showing, quite surprisingly, that lower-bound techniques for arithmetic circuits can be transformed into robust learning algorithms for this model.  The running time of their algorithm is $\poly(n,d,r)^\delta$ as well.

In the next section, we formally state our results,
followed by a detailed comparison with prior work.

\subsection{Our Results }

We consider the model of sums of powers of low degree sparse polynomials, denoted by  
\[
\Sigma^{[r]}\!\wedge^{[d]}\!\Sigma^{[s]}\!\Pi^{[\delta]}.
\]
Here, each bottom product gate has degree at most $\delta$. Above it, each addition gate computes an $s$-sparse polynomial. Then we take the $d$-th power of each sparse polynomial and sum these terms.
This model naturally includes the classical $\Sigma\!\wedge\!\Sigma$ (Waring) model as a special case.  
We give efficient deterministic hitting sets as well as a reconstruction algorithm for this class in the \emph{worst case}, without any genericity or average-case assumptions. Our results hold both for characteristic zero fields and for fields with polynomially large characteristic.

Our first result is a polynomial-size hitting set for $\Sigma^{[r]}\!\wedge^{[d]}\!\Sigma^{[s]}\!\Pi^{[\delta]}$
circuits, where the powering degree is quadratic in the \emph{rank}, namely, the top fan-in $r$. The following is a simplified version of \autoref{thm:hs-depth-4}.

\begin{theorem}\label{thm:main-pit}
    Let $n,d,r,s,\delta\in \N$ such that 
    $d=\Omega(r^2)$. Let $\F$ be a field of characteristic $p=0$ or $p\ge rd\delta   (s^2n+\delta)$. There is an explicit hitting set of size   $\poly(s,n,d)$ for the class of $\Sigma^{[r]}\!\wedge^{[d]}\!\Sigma^{[s]}\!\Pi^{[\delta]}$ circuits defined over $\F$.\footnote{Here and later, if $|\F|$ is not large enough then the evaluation points come from an extension field.}
\end{theorem}

To our knowledge, this is the first polynomial-size hitting set for any model of depth-$4$ circuits with unbounded top fan-in.

\autoref{thm:main-pit} is instrumental in obtaining a deterministic reconstruction algorithm.

\begin{theorem}\label{thm:main-rec}
    Let $n,d,r,s,\delta\in \N$ such that 
    $d\geq (r+1)^4\delta $. Let $\F$ be a field of characteristic $p=0$ or $p\ge rd\delta   (s^2n+\delta)$.
    %$p>4rd\delta n(s^2n + \delta)$.
    There exists a deterministic algorithm that, given black box access to a polynomial $f$ computed by a  $\Sigma^{[r]}\!\wedge^{[d]}\!\Sigma^{[s]}\!\Pi^{[\delta]}$ circuit, defined over $\F$, with bit complexity $B$, reconstructs $f$.  \begin{enumerate}
        \item When $p=0$, the running time is $\poly(n,s,d,B)$.
        \item When $p>rd\delta   (s^2n+\delta)$ and $|\F|=q$, the running time is $\poly(n,s,d,p,\log q)$.
    \end{enumerate}
    %In addition, there is a randomized algorithm running in time $\poly(n,s,r,d,\delta,B,\log p)$, when $p>2dr\delta$ .
    \end{theorem}

This is the first reconstruction algorithm for any model of depth-$4$ powering circuits with unrestricted top and bottom fan-in.

\subsection{Prior Work and Related Models}

In what follows, $\Sigma^{[r]}$ represents a layer of addition gates of fan-in at most~$r$, $\wedge^{[d]}$ represents raising to power~$d$, and $\Pi^{[\delta]}$ is a product gate of fan-in at most~$\delta$.  
Thus, $\Sigma^{[r]}\!\Pi^{[\delta]}\!\Sigma$ is the class of depth-$3$ circuits that consist of the sum of $r$ terms, each a product of at most~$\delta$ linear functions.  
A circuit is \emph{multilinear} if each monomial in every gate is multilinear.  
A circuit is \emph{set-multilinear} if the input is composed of disjoint sets $\vx = \vx_1 \sqcup \vx_2 \sqcup \cdots \sqcup \vx_\delta$, and each gate computes a polynomial in which each monomial contains at most one variable from each set.

\subsubsection{Polynomial Identity Testing for Small-Depth Circuits}

The breakthrough of Limaye, Srinivasan, and Tavenas \cite{limaye2025superpolynomial}, followed by Andrews and Forbes~\cite{andrews2022ideals}, established the first subexponential-size hitting sets for general constant-depth circuits.  
Sharper bounds are known for several restricted subclasses.

For \emph{sparse polynomials}, that is, polynomials computable by small $\Sigma\Pi$ circuits, deterministic hitting sets were obtained in the classical work of Ben-Or and Tiwari \cite{Ben-OrT88}, Grigoriev, Karpinski and Singer \cite{GrigorievKS90}, and Klivans and Spielman \cite{KS01}.  

Building on the white-box PIT algorithm of Raz and Shpilka~\cite{RazShpilka05}, Saxena~\cite{Saxena08diagonal} gave a deterministic polynomial-time PIT algorithm for $\Sigma^{[r]}\!\wedge^{[d]}\!\Sigma$ circuits.  
In the black-box model, \cite{forbes2014hitting,gurjar2016identity,GuoG20} constructed hitting sets of size $r^{\log\log r}\poly(d,n)$, which are superpolynomial in the width, but do not assume any relation between $d$ and $r$.

Rank-based  methods yield polynomial-size hitting sets for $\Sigma^{[r]}\!\Pi^{[d]}\!\Sigma$ circuits when the top fan-in $r$ is constant~\cite{DvirS07,KarninS11pit,KayalSaxena07,KayalSaraf09,SaxenaSesh13rank,SaxenaSesh12pit}.  
Agrawal et~al.~\cite{AgrawalSSS16} introduced the Jacobian hitting technique, obtaining polynomial-size hitting sets when the transcendence degree of the product terms is bounded.  
Their result also extends to certain classes of higher depths.

Using rank-based ideas, Peleg and Shpilka,  and Garg, Oliveira, and Sengupta~\cite{peleg2021polynomial,GargOS25a} gave polynomial-size hitting sets for $\Sigma^{[3]}\!\Pi\Sigma\!\Pi^{[O(1)]}$ circuits. Independently of \cite{GargOS25a}, Guo and Wang \cite{GuoWang25} gave a hitting set for $\Sigma^{[3]}\!\Pi\Sigma\!\Pi^{[O(1)]}$ circuits, under the additional restriction that in one of the terms the multiplicity of each irreducible factor is $1$. Their approach is novel and based on algebro-geometric techniques. However, the additional assumption simplifies the technical challenges encountered in earlier works.
Dutta, Dwivedi, and Saxena~\cite{DuttaD021} obtained quasipolynomial-size hitting sets for circuits whose top and bottom fan-ins are bounded by $\poly(\log n)$, using the newly developed DiDIL technique.  
Quasipolynomial and later polynomial-size hitting sets were also constructed for multilinear circuits with bounded top fan-in~\cite{KarninMSV13,SarafV18}.  
As in the depth-3 setting, polynomial-size constructions are known only when the top fan-in is bounded.

For sums of powers, i.e., $\Sigma^{[r]}\!\wedge^{[d]}\!\Sigma\!\Pi^{[\delta]}$ circuits, Forbes~\cite{Forbes15} constructed hitting sets of size $(nd)^{O(\delta\log r)}$.  
Dutta et~al.~\cite{DuttaD021} observed that existing techniques give hitting sets of size $s^{O(\log\log s)}$ for size-$s$ $\Sigma\!\wedge\!\Sigma\!\wedge$ circuits.  
These results do not restrict the top fan-in $r$, but the constructions remain superpolynomial in size.

To summarize, polynomial-size hitting sets are known for depth-2 circuits and for depth-3 and depth-4 models with bounded top fan-in.  
For depth-$4$ powering circuits, all known constructions are quasipolynomial (or slightly better but still superpolynomial) and hold without any restriction on the top fan-in.  

\autoref{thm:main-pit} gives the first polynomial-size construction, requiring only $d=\Omega(r^2)$.

\subsubsection{Reconstruction}

Shpilka \cite{Shpilka09-rec-sps} and  Karnin and Shpilka~\cite{KarninS09rec} gave the first reconstruction algorithms for $\Sigma^{[r]}\!\Pi\Sigma$ circuits, over small finite fields for $r=O(1)$.
Sinha~\cite{Sinha16,Sinha22} obtained randomized polynomial-time reconstruction algorithms for $\Sigma^{[2]}\!\Pi\Sigma$ circuits over $\mathbb{R}$ and $\mathbb{C}$, relying on Sylvester–Gallai type theorems. Saraf and Shringi~\cite{saraf2025reconstruction} extended this to $r=3$. More recently, Saraf, Shringi, and Varadarajan \cite{SarafSV25} solved the general $r=O(1)$ case, obtaining a quasipolynomial-time algorithm. 

For $\Sigma\!\wedge\!\Sigma$ circuits, the first reconstruction algorithm dates back to Sylvester, who gave an algorithm for binary forms provided $r \le \lfloor (d+2)/2 \rfloor$ in the generic\footnote{Being generic means belonging to a Zariski open subset.} case~\cite{sylvester1851lx}.  
Kayal~\cite{Kay12} gave a randomized reconstruction algorithm for the generic case when $r$ can be as large as $\binom{n+{d}/{2}-1}{{d}/{2}}$.  
When $d>2r$, Kayal's algorithm can be derandomized.

Bhargava, Saraf, and Volkovich~\cite{BhargavaSV21} gave a deterministic polynomial-time algorithm for multilinear and set-multilinear $\Sigma^{[r]}\!\Pi\Sigma$ circuits with top fan-in $r = O(1)$.  
In particular, this yields exact tensor decomposition for constant rank (with $r$ appearing in the exponent of the running time).  
Peleg, Shpilka, and Volk~\cite{peleg2022tensor} extended this to superconstant $r$, providing a randomized fixed-parameter algorithm with respect to $r$ (with poor dependence on the top fan-in).  
For set-multilinear depth-$3$ circuits, Bhargava and Shringi~\cite{BhargavaS25} further obtained a deterministic fixed-parameter algorithm running in time $2^{r^{O(1)}}\poly(n,d)$.  
While these works resolve the question for fixed-parameter settings, their dependence on $r$ remains superexponential or tower-type.

Garg, Kayal, and Saha~\cite{GKS20} studied circuits of the form  
$\Sigma^{[r]}\!\wedge^{[d]}\!\Sigma^{[s]}\!\Pi^{[\delta]}$.  
They obtained randomized algorithms under non-degeneracy assumptions in the noiseless case, with running time $\poly(n,r,d)^\delta$.  
In addition to the requirement that the input is generic, they assume $n>d^2$, $\delta \le O(\sqrt{(\log d)/(\log\log d)})$, and that the field has size at least $(nr)^{\Omega(\delta)}$.  
On the other hand, they allow $r$ to be as large as $n^{O(d/\delta^2)}$. 
An important aspect of their work is the method they introduced for translating certain algebraic lower-bound techniques to learning algorithms.  
Bafna et al.~\cite{bafna2022polynomial} extended the result of~\cite{GKS20} to handle polynomially large noise and to allow, roughly, $r=O(n^{2d/15})$.  
The subsequent work of Chandra et al.~\cite{ChandraGKMS24} generalized the method of~\cite{GKS20} and enabled transferring lower-bound techniques also to the noisy case, assuming natural conjectures on the largeness of singular values of certain random matrices.  
All these algorithms are randomized.

\autoref{thm:main-rec} provides a deterministic polynomial-time reconstruction algorithm in the worst-case noiseless setting, without assuming any non-degeneracy conditions.  

Our requirement $d=\Omega(r^4\delta)$ is similar in spirit to the requirement that $r \le n^{O(\delta)}$ from~\cite{GKS20,bafna2022polynomial,ChandraGKMS24}, since both conditions guarantee uniqueness of representation: ours for the worst case, and theirs for the generic case.

Additionally, \autoref{thm:main-rec} is the first polynomial-time reconstruction algorithm for any model of depth-4 circuits with unbounded top fan-in.

\subsection{Proof Overview}\label{sec:proof-overview}

Our proof of \autoref{thm:main-pit} is based on the ABC theorem for function fields, proved independently by Mason~\cite{mason} and Stothers~\cite{stothers1981polynomial}. We shall use an extension proved by Vaserstein and Wheland~\cite{vaserstein2003vanishing} (\autoref{thm:ABC}).  
Roughly speaking, the theorem shows that high-degree, pairwise independent univariate polynomials with relatively few distinct zeros are linearly independent.  
Therefore, if we start with $n$-variate polynomials that are high powers of sparse polynomials and project them to univariate polynomials in a way that preserves pairwise independence, the theorem guarantees that their sum cannot vanish identically.  
A simple interpolation then suffices to verify that the polynomial is nonzero.

The reconstruction algorithm (\autoref{thm:main-rec}) is more involved and requires several additional ideas.  
At a high level, we first reduce the problem to reconstructing in the univariate case and then lift the solution back to the $n$-variate setting. We next explain the ideas used in each of the settings.

\paragraph{Reconstruction in the univariate case.}
Consider a representation $f=\sum_{i=1}^{r}\alpha_i f_i(x)^d$ with monic $f_i$.  
Our approach is to find a differential operator of order $r$,
\[
L=\sum_{i=0}^{r}Q_i(x)\nabla^i,\qquad \text{where}\quad\nabla=\frac{d}{dx},
\]
such that $L(f)=0$.  
The existence of $L$ can be proved by considering the Wronskian $W(f,f_1^d,\ldots,f_r^d)$, where $W(g_1,\ldots,g_k)$ denotes the determinant of the $k\times k$ matrix whose $(i,j)$ entry is $\nabla^{i-1}g_j$.  
It is well known that the Wronskian of  polynomials vanishes if and only if they are linearly dependent, and this also holds for sufficiently large characteristic (see \autoref{thm:wronskian-lin-ind}).  
Expanding $W(f,f_1^d,\ldots,f_r^d)$ along the first column yields (after clearing the gcd) the desired operator $L$ and also provides an upper bound on the degrees of the $Q_i$.

The next crucial observation is that $\ker{L}$ is exactly the space spanned by $\{f_1^d,\ldots,f_r^d\}$.  
This again follows from the ABC theorem, in a slightly more general form (\autoref{cor:ABC}), which we also use to prove the uniqueness of $L$ (\autoref{cla:L-unique}).

Given this, we can compute the unique $L$ by solving a system of linear equations.

\medskip

Another important  observation is that each zero of each $f_i$ is also a zero of the leading coefficient $Q_r$ of $L$.  
In the theory of ODEs this is usually stated for poles of analytic functions, but for polynomials one can work with roots instead of poles, and the same argument applies in large characteristic (\autoref{cla:qr-zero}).  

We have thus reduced the problem to the following: we are given an $r$-dimensional space of polynomials (the kernel of $L$) that contains $r$ $d$-th powers of polynomials, whose zeros are known (but not their multiplicities), and we need to recover the underlying polynomials.  
A brute-force search over exponent vectors would require testing all $\binom{m+\delta}{\delta}$ possibilities, where $m=\deg(Q_r)$, but we will show a much more efficient algorithm that is polynomial in all parameters.

The starting point is to prove, by analyzing how the Wronskian factorizes, that any high-order root of the Wronskian must also be a root of one of the $f_i$ (in practice we work with irreducible factors rather than with individual roots).  
This enables us to prune the space $\Span\{f_1^d,\ldots,f_r^d\}$ and retain only the subspace of polynomials that vanish,  with high multiplicity, at some root of $Q_r$.  
We can think of this as building a depth-$\delta$ tree in which each node has $m$ children labeled by the roots of $Q_r$.  
Every path in the tree corresponds to a sequence of roots with associated multiplicities.  
For each node we construct the space of polynomials that vanish at these roots with given the multiplicities.  
At the leaves of the tree sit the $f_i^d$.  
The tree is still large, so we cannot check all its leaves. The idea is to perform a depth-first search and prune any branch that can be certified not to yield a new polynomial.

After the DFS algorithm finds the set $\{f_i^d\}$, we compute the coefficients $\alpha_i$ 
%\ynote{we don't require the HS} 
by solving a corresponding linear system.

% \paragraph{From bivariate to univariate.} \AS{rewrite, we don't do that any more}

% Assume now that we have an identity involving bivariate polynomials,
% \[
% f=\sum_{i=1}^{r}\alpha_i f_i(x,y)^d.
% \]
% Substituting any value $\beta$ for $y$ reduces the problem to the univariate case, in which we already know how to reconstruct the set of polynomials $\{f_i(x,\beta)\}$.  
% The main difficulty is that for different values of $\beta$ the reconstructed lists are not aligned - each reconstruction may return the $f_i$ in a different order.

% Our idea is to run a list-recovery procedure for Reed-Solomon codes, over enough evaluation points $\beta$, to recover the coefficients of $x$ as polynomials in $y$.  
% If $f_i(x,y)=\sum_j f_{i,j}(y)x^j$, then learning $f_i(x,\beta)$ for various $\beta$ gives access to the sets $\{f_{i,j}(\beta)\}_j$.  
% We can run list recovery on these sets to obtain the polynomials $\{f_{i,j}(y)\}_j$.  
% To align the different lists, we choose a special value $\beta$ that serves as a fingerprint for each $f_i$, allowing us to match the corresponding coefficients across different $\beta$.

\paragraph{Multivariate reconstruction.}

At a high level, we would like to project the blackbox multivariate polynomial to a univariate polynomial, apply the univariate reconstruction algorithm, and then, given the resulting univariate polynomials, recover the original polynomials.

To achieve this, we first observe that the projection must preserve non-associateness.\footnote{Polynomials are associate if and only if they are linearly dependent.} This can be ensured relatively easily using known tools, such as the Klivans-Spielman generator \cite{KS01}.
However, we cannot allow the degree of the projected polynomial to increase. This forces us to use projections of the form $\vx=\vu+ t\cdot (\vu - \vv)$, where $\vu,\vv$ will come from appropriate sets.

We are then left with the task of recovering the original polynomials from their projected versions. Each such projection corresponds to restricting a multivariate polynomial to a line. Consequently, reconstructing the polynomial requires its values on sufficiently many different lines, and hence we must consider many distinct projections.

At this point, another difficulty arises. Each univariate reconstruction produces a list of $r$ polynomials, but these lists are not ordered. In particular, they are not aligned across different projections. Therefore, we must first align the lists in order to obtain black-box access to the projections of a single polynomial. This is accomplished by showing the existence of a suitable `anchor' point $\vu$ that enables the alignment of all projections of the form  $\vu+ t\cdot (\vu - \vv) \to \vx$, where $\vv$ ranges over a sufficiently large (and structured) set. We show that such a good `anchor' point exists in any hitting set for $\Sigma^{[2]}\mathord{\wedge}^{[d]}\Sigma^{[2s]}\Pi^{[\delta]}$  circuits.

Finally, we show that once alignment is achieved, the available information, obtained from the different $\vv$, suffices to reconstruct the original polynomials. This process is reminiscent of erasure correction in error-correcting codes, except that our target objects (i.e., sparse polynomials) do not form a vector space. Nevertheless, the reconstruction can be carried out using a robust version of the Klivans–Spielman hitting set \cite{KS01}.

% This step is comparatively simpler and uses standard techniques.  
% We employ the Klivans–Spielman generator to project all variables except one (each time a different variable survives) to powers of a new variable $y$, producing bivariate polynomials.  
% The projection is designed to preserve sparsity.  
% For each projection we reconstruct the resulting bivariate polynomials as described above, and then use ideas from sparse-polynomial reconstruction to recover the original multivariate polynomials.

% \subsection{Open Problems}

% It is natural to ask whether there are less strict parameter regimes for $d$. In particular, our \autoref{Theorem2} requires that the size of $d$ be related to $\delta$. Does the size of $d$ need to be related to the individual degree bounds of the bottom polynomials for efficient reconstruction? In other words, an interesting problem is to derive similar results to ours but where the assumption on the size of $d$ relates only to $r$.

\subsection{Organization}

In \Cref{sec:prelim} we introduce the basic mathematical tools. These include the ABC theorem (\Cref{sec:abc}), the robust Klivans-Spielman generator (\Cref{sec:KS}), the Wronskian and its properties (\Cref{sec:wronski}), and results on deterministic factorization of univariate polynomials (\Cref{sec:factorization}). 

We present our hitting sets in \Cref{sec:pit}. In \Cref{sec:univariate-rec} we give our univariate reconstruction algorithm and in \Cref{sec:multi} we give the algorithm for the multivariate case.

\section{Preliminaries}
\label{sec:prelim}

For an integer $r\in \N$ we denote $[r]=\{1,\ldots,r\}$ and $[[r]]=\{0,1,\ldots,r\}$.

We say that two nonzero polynomials $f,g\in\F[\vx]$ are associate if there exists a nonzero scalar $\alpha\in\F$ such that $g(\vx)=\alpha f(\vx)$. We say that a set of polynomials is non-associate if no two polynomials in the set are associate. This is equivalent to saying that the polynomials in the set are pairwise linearly independent.

\subsection{The ABC theorem for function fields}\label{sec:abc}

An important tool in all proofs is the ABC theorem for function fields, known as the Mason-Stothers theorem \cite{mason,stothers1981polynomial}.

We state the following strengthening  due to  Vaserstein and Wheland \cite{vaserstein2003vanishing}. 
For a polynomial $h\in\F[x]$, let $\nu(h)$ denote the number of \emph{distinct} roots of $h$ over the algebraic closure $\overline{\F}$.

While most of the results in \cite{vaserstein2003vanishing} are stated for fields of characteristic zero, an inspection of the proofs in   \cite[Section 3]{vaserstein2003vanishing} reveals that they also hold for fields of large enough characteristic.\footnote{Vaserstein and Wheland were primarily interested in characteristic $0$ or fixed characteristic independent of the degrees and number of polynomials (see \cite[Section 5]{vaserstein2003vanishing}). }

\begin{theorem}[\cite{vaserstein2003vanishing}, Theorem~2.2(a)]\label{thm:ABC}
 Let $D, r\in \N$ such that $r\geq 2$. Let $\F$ be a field of characteristic $p$ such that $p=0$ or $p>rD$. Let $h_1,\ldots,h_r\in\F[x]$ satisfy
\[
h_1+\cdots+h_r = h_0,\qquad \gcd(h_1,\ldots,h_r)=1,
\]
where $\max_i\deg{h_i}\leq D$, not all $h_i$ are constant and no nonempty subsum of $h_1,\ldots,h_r$ vanishes. 
Then
\[
\deg(h_0) < (r-1)\!\left(\sum_{i=0}^r \nu(h_i)\right).
\]
\end{theorem}

In other words, \autoref{thm:ABC} asserts that if $h_0$ has ``high'' degree while each $h_i$ has ``few'' distinct roots, then $h_0$ cannot be expressed as a linear combination of the other $h_i$’s. 
We will use the following corollary. 

We note that we require an upper bound on $\deg(h_i)$ only in the case of positive characteristic.

\begin{corollary}\label{cor:ABC}
Let $r,\eta,\delta,e\in \N$ such that $r\geq 2$.
Let $\F$ be a field of characteristic $p$ such that $p=0$ or $p>r(\eta+e\delta)$. Let $P_0(x)^e,P_1(x)^e,\ldots,P_r(x)^e\in \F[x]$ be pairwise linearly independent polynomials, such that for all $i$, $\deg(P_i)\leq \delta$. 
Let $g_0(x),\ldots,g_r(x)\in\F[x]$ be additional polynomials, such that  for all $i$, $\deg(g_i)\leq \eta$. 
Assume that
\[
e \;\ge\; (r^2+r)\,(\eta +1).
\]
Then the identity
\[
\sum_{i=0}^r g_i(x)\,P_i(x)^e = 0
\]
holds if and only if each $g_i(x)$ is identically zero.
\end{corollary}

\begin{proof}
The “if” direction is immediate. For the converse, suppose
\[
\sum_{i=0}^r g_i P_i^e = 0
\]
is a nontrivial identity. Passing, if necessary, to a nontrivial relation of minimal support, we may assume that every $h_i := g_i P_i^e$ in the relation is nonzero, and hence each corresponding $g_i$ is nonzero. Set
\[
H(x)=\gcd(h_0,\ldots,h_r),\qquad P(x)=\gcd(P_0,\ldots,P_r) .
\] 
%\ynote{$a:=deg(P)$ was never used so I removed it}

\begin{claim}
There exists $G(x)\in\F[x]$ such that $H(x)=P(x)^e G(x)$ and $\deg G \le (r+1)\eta$.
\end{claim}

\begin{proof}
Clearly $P^e$ divides each $h_i$, so $P^e$ divides $H$. As $\gcd((P_0/P)^e,\ldots,(P_r/P)^e)=1$, we get that
\[
\gcd\big(g_0(P_0/P)^e,\ldots,g_r(P_r/P)^e\big)\quad \text{divides}\quad \Pi_{i=0}^r g_i.
\]
Thus $G\coloneq H/P^e$ divides $\Pi_i g_i$, so
\[
\deg(G)=\deg(H/P^e)\ \le\ \sum_i \deg g_i\ \le\ (r+1)\eta.\qedhere
\]
\end{proof}

Let $S := \{i:\, g_i \not\equiv 0\}$ be the support, and choose $t\in S$ such that
\[
\nu_{\mathrm{sup}} := \max_{i\in S} \nu(P_i/P)= \nu(P_t/P) .
\]
Since the polynomials are non-associate, $\nu_{\mathrm{sup}}\ge 1$.
Define $\tilde h_i := h_i/H$. 
Then the equality $\sum_{i=0}^r h_i=0$ holds if and only if $\sum_{i=0}^r \tilde h_i=0$, and moreover
\[
\gcd(\tilde h_0,\ldots,\tilde h_r)=1.
\]

By \autoref{thm:ABC}, if $\sum_{i=0}^r \tilde h_i=0$ and not all $\tilde h_i$ vanish, then
\begin{align*}
 e\,\nu_{\mathrm{sup}}-\deg(G)&= e\nu(P_t/P)-\deg(G)\\
 &\le \deg((P_t/P)^e)-\deg(G)
 \\& \le \deg(\tilde h_t)\\
 &< (r-1)\sum_{j=0}^r \nu(\tilde h_j)\\ &\le (r-1)\sum_{j=0}^{r}\nu(g_j (P_j/P)^e)\\
 &\le (r-1)(r+1)\big(\eta+\nu_{\mathrm{sup}}\big).
\end{align*}
Hence
\[
e < (r^2+r)(\eta+1),
\]
contradicting the assumption on $e$. 
\end{proof}

We did not attempt to optimize the parameters in \autoref{cor:ABC}. In particular, for the special case where all $g_i$ are scalars a tighter result is known.

\begin{theorem}[{\cite[Corollary 3.8]{vaserstein2003vanishing}}]\label{thm:abc}
     Let $r,d,\delta\in \N$ such that $r\geq 2$. Let $\F$ be of characteristic $p=0$ or $p>rd\delta$. Let $f_0(x),\ldots,f_r(x)\in\F[x]$ be non-associate polynomials of degree at most $\delta$, not all the $f_i$ are constant. If $(r-1)^2\leq d+1$ then $\sum_{i=0}^{r}f_i^d(x) \neq 0$.
\end{theorem}

\subsection{The Klivans–Spielman generator}\label{sec:KS}

We recall the generator construction of Klivans and Spielman~\cite{KS01}, which converts a sparse multivariate polynomial into a univariate polynomial while preserving its structure.

For a prime number $q$, an integer $k<q$ let $\Psi_{k,q}:\F[\vx]\mapsto \F[y]$ be defined as 
\begin{equation}\label{eq:psi-k-q}
\Psi_{k,q}(x_i)=
y^{(k^{i-1}\bmod q)},    
\end{equation}
with the natural extension to monomials and polynomials. Namely,  for an exponent vector $\ve=(e_1,\ldots,e_n)$
\[
\Psi_{k,q}(\vx^\ve)= y^{\sum_{i} e_i\left(k^{(i-1)}\bmod q\right)}.
\]
Clearly, 
\begin{equation}\label{eq:deg-psi}
 \deg(\Psi_{k,q}(\vx^\ve))\leq  (q-1)\sum_{i}e_i=(q-1)\deg(\vx^\ve) \ .   
\end{equation}

\begin{theorem}[Corollary of Klivans–Spielman~{\cite[Lemma~1]{KS01}}]\label{thm:PIT:sparse:KS}
Let $f\in\F[\vx]$ be a nonzero $n$-variate polynomial, such that $f$ has at most~$s$ monomials and  $\deg(f)\leq \delta$.
Let $q$ be a prime satisfying $q\ge \max(\delta,\binom{s}{2}n)+1$.
Then, for all but at most $\binom{s}{2}(n-1)$ values of $k\in[q-1]$, the univariate polynomial $\Psi_{k,q}(f)$
has the same number of monomials as $f$. Furthermore, $\deg(\Psi_{k,q}(f))\leq \delta q$. Moreover, for all but at most $(s-1)(n-1)$ values of $k\in[q-1]$, the univariate polynomial $\Psi_{k,q}(f)$ is nonzero.
\end{theorem}

\begin{corollary}\label{cor:ks-lin-ind}
Let $f,g\in\F[\vx]$ be non-associate polynomials of individual degrees at most~$\delta$, each with at most~$s$ monomials.
Let $q$ be a prime larger than $\max(\delta,s^2n)+1$.
Then, for all but at most $s(s-1)(n-1)$ values of $k\in[q-1]$, the univariate specializations
$\Psi_{k,q}(f),\Psi_{k,q}(g)$
are linearly independent over $\F$.
\end{corollary}

\begin{proof}
By \autoref{thm:PIT:sparse:KS}, for all but $s(s-1)(n-1)$ values of $k\in[q-1]$, the specializations preserve the supports of $f$ and $g$.
If $f$ and $g$ have different supports, the property is immediate.
Otherwise, they share the same set of monomials but have linearly independent coefficient vectors, and this linear independence is preserved under the substitution.
\end{proof}

We shall use the Klivans-Spielman generator in order to create a robust interpolation set for sparse polynomials. That is, a set of points such that knowing the value of an $s$-sparse polynomial on $1-\epsilon$ fraction of the points in the set, allows efficient recovery of the polynomial.

The idea is quite simple given \Cref{thm:PIT:sparse:KS}. Assume first the case of a characteristic zero field. Consider the modified generator
 \[\Psi_{j,k,q,\lambda}(x_i)=
 \begin{cases}
 \lambda \cdot y^{(k^{j-1}\bmod q)} & i=j,\\
 y^{(k^{i-1}\bmod q)} & i\ne j.
 \end{cases},
 \]
 Observe that each monomial $\vx^\ve$ is mapped to
 \[
 \Psi_{j,k,q,\lambda}(\vx^\ve)= \lambda^{e_j}\cdot y^{\sum_{i} e_i\left(k^{(i-1)}\bmod q\right)}.
 \]
Thus, if we know both $\Psi_{k,q}(\vx^\ve)$ and $\Psi_{j,k,q}(\vx^\ve)$ then we can recover the exponent of $x_j$ from each of the monomials. To adapt this to finite fields, we note that all we need $\lambda$ to satisfy is that its order  in the multiplicative group $\F^{*}$,  is larger than $\delta$. Since the order of the group is $|\F|-1$, we may have to pick $\lambda$ from an extension field. Similarly, to interpolate a polynomial of degree $\delta q$ we would need a field of size at least $\delta q$. This leads to the following construction of a robust interpolating set.

For the construction we shall need the following specialization of $\Psi_{j,k,q,\lambda}$ which we denote with $\Psi_{j,k,q,\lambda}[\alpha]:\F[\vx]\to\F^n$, for any $\alpha\in \F$:
\[\Psi_{j,k,q,\lambda}[\alpha]_i=
 \begin{cases}
 \lambda \cdot \alpha^{(k^{j-1}\bmod q)} & i=j,\\
 \alpha^{(k^{i-1}\bmod q)} & i\ne j.
 \end{cases}
 \]
In other words, $\Psi_{j,k,q,\lambda}[\alpha]$ is obtained by substituting $\alpha$ into $y$ in the vector $(\Psi_{j,k,q,\lambda}(x_1),\ldots,\Psi_{j,k,q,\lambda}(x_n))$.
We also denote with $\Psi_{k,q}[\alpha]:\F[\vx]\to\F^n$ the map
\[\Psi_{k,q}[\alpha]_i =  \alpha^{(k^{i-1}\bmod q)}\ ,\]
which is obtained by substituting $\alpha$ into $y$ in the vector $(\Psi_{k,q}(x_1),\ldots,\Psi_{k,q}(x_n))$.

The following construction is far from being optimal in terms of its size, but it is relatively simple to describe.

\begin{construction}[Robust interpolating set for sparse polynomials]\label{construction:robust-ks}
Let $n,s,\delta\in \N$ and $\epsilon>0$ ($\epsilon$ may depend on $n,s,\delta$).
Let 
\[2\delta s^2n^2/\epsilon <q < 4\delta s^2n^2/\epsilon\] be a prime number.
Let $\cA_{n,\delta,q,\epsilon}\subset\F$ be a set of size  \[|\cA_{n,\delta,q,\epsilon}| =  \ceil{2n\delta q/\epsilon}.\] If $|\F|$ is too small then we pick $\cA_{n,\delta,q,\epsilon}$ from an extension field $\mathbb{E}$ of size $|\F|^2\leq |\mathbb{E}|\leq 2n\delta q/\epsilon\cdot |\F| $. 

For $\lambda\in \F$ (or $\lambda\in \mathbb{E}$ if needed) denote 
\begin{equation*}
    \begin{aligned}
        \cS_{n,\delta,q,\epsilon,\lambda} &= \\ \bc{\Psi_{j,k,q,\lambda}[\alpha] \mid j\in [n],  k\in [q-1], \alpha\in\cA_{n,\delta,q,\epsilon}}\ &\cup \  \bc{\Psi_{k,q}[\alpha] \mid k\in [q-1], \alpha\in\cA_{n,\delta,q,\epsilon}} \ .
    \end{aligned}
\end{equation*}
Observe that 
\[| \cS_{n,\delta,q,\epsilon,\lambda}|=(n+1)(q-1)|\cA_{n,\delta,q,\epsilon}|= O\left(\delta n^2q^2/\epsilon  \right)= O\left(\delta^3 n^6s^4/\epsilon^3  \right)\ .\]
\end{construction}
We first prove that $\cS_{n,\delta,q,\epsilon,\lambda}$ is a good hitting set for sparsity $2s$. 
\begin{lemma}\label{lem:const-good-pit}
    Assume the notation of \Cref{construction:robust-ks}. Let $f\in\F[x]$ be an $n$-variate, $2s$-sparse polynomial of degree $\deg(f)=\delta$. Then, $f$ vanishes on at most a fraction of $\frac{\epsilon}{n}$ of the points in $ \cS_{n,\delta,q,\epsilon,\lambda}$.
\end{lemma}
\begin{proof}
By the ``moreover'' part of \Cref{thm:PIT:sparse:KS}, the number of $k\in [q-1]$ for which $\Psi_{k,q}(f)=0$ is at most $(2s-1)(n-1)$. For the rest of the $k$'s, since $\deg(\Psi_{k,q}(f))< \delta q$, it has at most $\delta q$ zeroes. Thus, the total number of points in $ \cS_{n,\delta,q,\epsilon,\lambda}$ on which $f$ vanishes is at most
\[(2s-1)(n-1)\cdot (n+1)|\cA_{n,\delta,q,\epsilon}| + q\cdot q\delta< \left(\frac{2sn}{q}+\frac{\epsilon}{n^2}\right)| \cS_{n,\delta,q,\epsilon,\lambda}|<\frac{\epsilon}{n}| \cS_{n,\delta,q,\epsilon,\lambda}|\ .\qedhere\]
\end{proof}
We next prove that we can efficiently reconstruct from evaluations on a $(1-\epsilon)$ fraction of the points in $\cS_{n,\delta,q,\epsilon,\lambda}$.
\begin{theorem}\label{thm:robust-ks}
    Let $n,s,\delta\in \N$ and $0<\epsilon<\frac{1}{50 n}$. Let $\delta s^2n<q < 2\delta s^2n$ be a prime number. Assume that $\F$ is large enough to allow \Cref{construction:robust-ks}. Let $\lambda\in\F^{*}$ be an element whose multiplicative order is larger than $\delta q$ (if needed, pick $\lambda \in \mathbb{E}$ for an appropriate extension field $\mathbb{E}$).

    There is an algorithm that with the following guarantee: 
    Let $f\in \F[\vx]$ be an $n$-variate, $s$-sparse polynomial of degree $\deg(f)=\delta$. Denote with $B$ the bit-complexity of $f$'s coefficients.\footnote{If $\F$ is a finite field then $B=s\log|\F|$.}
    Let $\cS'\subseteq \cS_{n,\delta,q,\epsilon,\lambda}$ be any subset of size $|\cS'|\geq (1- {\epsilon} )|\cS_{n,\delta,q,\epsilon,\lambda}|$. Given the set of evaluations
    \[\bc{\left(\vbeta,f(\vbeta) \right) \mid \vbeta\in \cS'} \ , \]
    the algorithm runs  in time $\poly(n,s,\delta,1/\epsilon,B)$ and returns $f$.  
\end{theorem}

\begin{proof}
    For each $j\in [n]$, $k\in [q-1]$ denote
     \[\cE_{j,k,\lambda}:= \bc{\left(\vbeta,f(\vbeta) \right) \mid \vbeta= \Psi_{j,k,q,\lambda}[\alpha] , \text{ for some } \alpha\in\cA_{n,\delta,q,\epsilon}} \cap \cS' \ .\]
     Similarly, for each   $k\in [q-1]$ denote
     \[\cE_{k}:= \bc{\left(\vbeta,f(\vbeta) \right) \mid \vbeta= \Psi_{k,q}[\alpha] , \text{ for some } \alpha\in\cA_{n,\delta,q,\epsilon}} \cap \cS' \ .\]
     By our assumption, and definition of $\cS_{n,\delta,q,\epsilon,\lambda}$, we have
     \begin{equation*}
          (1- {\epsilon} )|\cS_{n,\delta,q,\epsilon,\lambda}| \leq |\cS'| = \sum_{j\in [n]}\sum_{k\in [q-1]}  |\cE_{j,k,\lambda}| + \sum_{k\in [q-1]}  |\cE_{k}|\ .
     \end{equation*}
     
    Call $k$ good if for every $j\in [n]$, it holds that $|\cE_{j,k,\lambda}|> \delta q$ and in addition, $|\cE_{k}|> \delta q$.
    \begin{claim}\label{cla:k-good}
        At most $(q-1)/4$ of the points  $k\in [q-1]$  are not good.
    \end{claim}
    \begin{proof}
        If $k$ is not good then
        \begin{align*}
              |\cE_{k}|+ \sum_{j=1}^{n} |\cE_{j,k,\eta}| &\leq \delta q +n|\cA_{n,\delta,q,\epsilon}|\\ &<  (n+1)|\cA_{n,\delta,q,\epsilon}|\left(1-\frac{1}{n+1} + \frac{\epsilon}{2n^2}\right) < \frac{|\cS_{n,\delta,q,\epsilon,\lambda}|}{q-1}\left(1- \frac{1}{4n}\right)\ ,
        \end{align*}
        where we have used $|\cA_{n,\delta,q,\epsilon}| =  \ceil{2n\delta q/\epsilon}$. Thus, if there are $b$ bad $k$'s, then we would have 
         \begin{equation*}
          (1- {\epsilon} )|\cS_{n,\delta,q,\epsilon,\lambda}| \leq |\cS'| = \sum_{k\in [q-1]} \left(|\cE_{k}|+ \sum_{j=1}^{n} |\cE_{j,k,\eta}|\right) \ < |\cS_{n,\delta,q,\epsilon,\lambda}|\left( 1-\frac{b}{4nq} \right).
     \end{equation*}
     Hence $b<4\epsilon n q<(q-1)/4$.
    \end{proof}
     
    By our choice of $q$, \Cref{thm:PIT:sparse:KS} guarantees that at least half of all $k\in [q-1]$ satisfy that  the univariate polynomial $\Psi_{k,q}(f)$ has the same number of monomials as $f$. Denote with $K$ the set of all such $k\in [q-1]$. Then $|K|\geq (q-1)/2$.
    \Cref{cla:k-good} implies that there is some $k\in K$ which is good. We call such a $k$ a good decoding point.

    % Denote with $K$ the set of all such $k\in [q-1]$. Then $|K|\geq (q-1)/2$. 

    % For every $k\in K$ denote
    % \[\cE_{j,k,1}:= \bc{\left(\vbeta,f(\vbeta) \right) \mid \vbeta= \Psi_{j,k,q,1}[\alpha] , \text{ for some } \alpha\in\cA} \cap \cS' \ ,\]
    % and
    % \[\cE_{j,k,\lambda}:= \bc{\left(\vbeta,f(\vbeta) \right) \mid \vbeta= \Psi_{j,k,q,\lambda}[\alpha] , \text{ for some } \alpha\in\cA} \cap \cS'\ .\]
    
    % \begin{claim}\label{cla:k-interpolate}
    %     There is an algorithm that given the evaluation of $f$ on the points in $\cE_{k}\cup\left(\bigcup_{j\in [n]} \cE_{j,k,\lambda}\right)$, for a good decoding point $k$, runs in time $\poly(n,s,\delta,B)$ and reconstructs $f$.
    % \end{claim}
%    \begin{proof}
        Let $k$ be a good decoding point.
        Denote $f=\sum c_{\ve}\vx^{\ve}$. Then, for every $j\in [n]$ we have
        \[ \Psi_{k,q}(f)=\sum c_{\ve}y^{\sum_{i} e_i\left(k^{(i-1)}\bmod q\right)} \]
        and
        \[ \Psi_{j,k,q,\lambda}(f)=\sum c_{\ve}\lambda^{e_j} y^{\sum_{i} e_i\left(k^{(i-1)}\bmod q\right)} \ .\]
        By choice of $k$, the number of monomials in each of these polynomials equals that of $f$. Furthermore, both polynomials have the same set of monomials and they differ only by their coefficients. Moreover, by \eqref{eq:deg-psi}, they have degree smaller than $\deg(f)\cdot q=\delta q$.
        
        Thus, the evaluation points in $\cE_{k}$ give us enough information to interpolate $ \Psi_{k,q}(f)$. Similarly, the evaluation points in $\cE_{j,k,\lambda}$ give us enough information to interpolate $ \Psi_{j,k,q,\lambda}(f)$.

        Since they have the same set of monomials, we can infer for each monomial $y^{\sum_{i} e_i\left(k^{(i-1)}\bmod q\right)}$ the term $\lambda^{e_j}$. As the order of $\lambda$ is larger than $\delta$, and in particular larger than $e_j$ we can easily infer $e_j$ from $\lambda^{e_j}$.  Since we can do that for every coordinate $j$, we can infer the coefficient vector $\ve$ corresponding to the monomial $y^{\sum_{i} e_i\left(k^{(i-1)}\bmod q\right)}$. Given the coefficients $c_\ve$ that we found, we have the complete representation of $f$.
%    \end{proof}
\end{proof}

\subsection{Wronskian and Differential Operators over Finite Fields}\label{sec:wronski}
We recall the definition of the \emph{Wronskian}.\footnote{The Wronskian can be defined for any collection of functions that are sufficiently differentiable on an interval; here we consider only polynomials.}
Let $g_1,\ldots,g_n\in \F[x]$. Their Wronskian is the determinant of the $n\times n$ matrix whose $(i,j)$-entry, for $i\in[[n-1]]$ and $j\in[n]$, is $\nabla^{(i)} g_j$, where
\[
\nabla^{(i)} := \frac{d^i}{dx^i}.
\]
\begin{equation}\label{eq:wronskian-g}
W(g_1,\ldots,g_n)\coloneq 
\det\!\begin{pmatrix}
    g_1(x) & g_2(x) & \cdots & g_n(x)\\
    \nabla g_1(x) & \nabla g_2(x) & \cdots & \nabla g_n(x)\\
    \vdots & \vdots & \ddots & \vdots\\
    \nabla^{(n-1)} g_1(x) & \nabla^{(n-1)} g_2(x) & \cdots & \nabla^{(n-1)} g_n(x)
\end{pmatrix}\in \F[x] .
\end{equation}
We treat $\nabla$ as a linear operator on the space of polynomials, so it is well defined over any field.

The Wronskian is a classical object in the theory of differential equations (see, e.g., \cite[\S3.4]{ODE}), but it also plays an important role for polynomials over finite fields. The key fact is that if the characteristic of the field is sufficiently large, then the Wronskian vanishes if and only if the polynomials are linearly dependent. For completeness, we provide the proof of the following theorem in \autoref{app:wronskian}. 

\begin{restatable}{theorem}{thmwronskilinind}\label{thm:wronskian-lin-ind}
Let $g_1, \dots, g_n$ be polynomials over a field $\F$ of characteristic $p$.  
If the maximum degree of any $g_i$ is $d$ and $p=0$ or $p > d$, then  
$g_1, \dots, g_n$ are linearly dependent over $\F$ if and only if their Wronskian
$W(g_1, \dots, g_n)$ is identically zero.    
\end{restatable}

The next theorem, an analogue of the fundamental theorem of linear homogeneous ODEs for polynomials over arbitrary fields, states that the solution space of an operator of order $r$ has dimension at most $r$.

\begin{theorem}\label{thm:ode}
Let $\F$ be a field of characteristic $p$, and let\footnote{Since $x$ and $\nabla$ do not commute, we denote the noncommutative polynomial ring of differential operators with coefficients in $\F[x]$ by $\F[x]\langle\nabla\rangle$.}
\[
A \;=\; \sum_{i=0}^{r} Q_i(x)\,\nabla^{i} \;\in\; \F[x]\langle\nabla\rangle
\]
be a nonzero operator of order $r$, with $Q_r\not\equiv 0$.
Fix $D\in\N$ and assume $p=0$ or $p> D$.
Then the space
\[
\mathcal S_{\le D}\;=\;\{\,g\in\F[x]:\ \deg g\le D,\ A(g)=0\,\}
\]
has $\dim_\F \mathcal S_{\le D}\le r$. 
\end{theorem}

\begin{proof}
Since $|\F|>D$, there exists $a\in\F$ with $Q_r(a)\neq 0$. Define the translated operator
\[
\tilde{A} \;=\; \sum_{i=0}^{r}  \tilde{Q}_i(x)\,\nabla^i,
\qquad  \tilde{Q}_i(x):=Q_i(x+a),
\]
and the translated polynomial $\tilde{g}(x):=g(x+a)$. 
Since ordinary derivatives commute with translation, we have
$A(g)=0$ if and only if $\tilde{A}(\tilde{g})=0$.
Write $\tilde{g}(x)=\sum_{i=0}^{D} c_i x^i$, where $c_0,\ldots,c_D\in\F$ are indeterminates. Then
\[
\nabla^k \tilde{g}=\frac{d^k \tilde{g}}{dx^k}(0) \;=\; k!\, c_k \qquad (0\le k\le D).
\]
Consider the identities $\nabla^t\left(\tilde{A}(\tilde{g})\right)(0)=0$ for $t=0,1,\ldots,D-r$. 
By Leibniz’s rule
\[
0 \;=\; \nabla^t\left(\tilde{A}(\tilde{g})\right)(0)
 \;=\; \sum_{i=0}^{r} \sum_{j=0}^{t} \binom{t}{j}\,\nabla^{t-j} {\tilde{Q}}_i(0)\,\nabla^{i+j} {\tilde{g}}_i(0),
\]
hence, using $\nabla^{i+j} {\tilde{g}}_i(0)=(i+j)!\,c_{i+j}$,
\begin{equation}\label{eq:triangular-recurrence-corrected}
\sum_{i=0}^{r} \sum_{j=0}^{t} \binom{t}{j}\, \nabla^{t-j} {\tilde{Q}}_i(0)\,(i+j)!\, c_{i+j} \;=\; 0,
\qquad t=0,1,\ldots,D-r.
\end{equation}
In \eqref{eq:triangular-recurrence-corrected}, the term with $(i,j)=(r,t)$ equals
\[
\binom{t}{t}\,{\tilde{Q}}_r (0)\,(r+t)!\, c_{r+t} \;=\; Q_r(a)\,(r+t)!\, c_{r+t},
\]
while all other terms involve only $c_k$ with $k\le r+t-1$. 
Because $p=0$ or $p>D$, each factorial $(r+t)!$ is invertible in $\F$ for $0\le t\le D-r$, 
and $Q_r(a)\ne 0$ by construction. Thus, the coefficient of $c_{r+t}$ is nonzero. 
Therefore, for each $t$, we can solve uniquely for $c_{r+t}$ in terms of $c_0,\ldots,c_{r+t-1}$. 
Inductively, all $c_r,\ldots,c_D$ are determined by the initial block $(c_0,\ldots,c_{r-1})\in\F^r$, 
so $\dim_\F \mathcal S_{\le D}\le r$.
\end{proof}

\subsection{Polynomial Factorization }\label{sec:factorization}

We shall need the following result on factorization of polynomials. The first result is the famous LLL algorithm for factorization over $\Q$. 

\begin{theorem}[Factorization over $\Q$, \cite{LLL}]\label{thm:lll}
    Let $f\in \Q[x]$ have degree $d$ and bit complexity $B$. Then, there is a deterministic algorithm that outputs all irreducible factors of $f$ whose running time is $\poly(d,B)$.
\end{theorem}

For the following theorem, see, e.g. \cite[Section 9]{GathenShoup1992-deterministic-factor}.

\begin{theorem}[Factorization over Finite Fields]\label{thm:finite-factor}
    Let $\F_q$ be a field of size $q=p^k$ and characteristic $p>0$.
    Let $f(x)\in \F_q[x]$ be a polynomial of degree $d$. Then, there is a deterministic algorithm that outputs all irreducible factors of $f$ in time $\poly(p,\log q,d)$.
\end{theorem}

%\anote{cite results on bivariate factorization}

%%%%%%%%%%%%%%%%%%%%%

\section{\texorpdfstring{Hitting set for $\Sigma^{[r]}\mathord{\wedge}^{[d]}\Sigma^{[s]}\mathord{\Pi}^{[\delta]}$ circuits with $r^2\le d$}{Hitting set for Sigma^{[r]} and^{[d]} Sigma Pi circuits with $r^2\le d$}}\label{sec:pit}

In this section we give a simple polynomial-size hitting set for $\Sigma^{[r]}\mathord{\wedge}^{[d]}\Sigma^{[s]}\mathord{\Pi}^{[\delta]}$ circuits when $d=\Omega(r^2)$.

    \begin{theorem}\label{thm:gen-4}
         Let $r,d,s,n,\delta\in \N$ such that  $(r-1)^2\leq d+1$.  Let $q$ be a prime such that  $q\geq \max(\delta,r^2s^2n)+1$ and let $t\geq q/\epsilon$. Let $\F$ be a field of characteristic $p$ such that $p=0$ or $p>rd\delta q $.
        Let $f(\vx)\in\F[\vx]$ be a nonzero polynomial that is computable by a  $\Sigma^{[r]}\mathord{\wedge}^{[d]}\Sigma^{[s]}\mathord{\Pi}^{[\delta]}$ circuit.  Then, except for at most $r^2s^2n$ values of $k\in [t]$, the specialization 
        $\Psi_{k,q}(f)$ defined as in \autoref{thm:PIT:sparse:KS} is nonzero.
    \end{theorem}

    \begin{proof}
        Let $f=\sum_{i=1}^{r} f_i^d$, such that each $f_i$ has degree at most $\delta$ and at most $s$ monomials, and any two are non-associate. Let $\Psi_{k,q}(f_i)$ be as in \autoref{cor:ks-lin-ind}. Then, by a simple application of the union bound, except for at most $r^2s^2n$ values of $k$, the polynomials $\Psi_{k,q}(f_1),\ldots,\Psi_{k,q}(f_r)$ are non-associate of degree smaller than $\delta  q$. Fix a $k$ that guarantees such pairwise linear independence. \autoref{thm:abc} guarantees that $\Psi_{k,q}(f)=\sum_{i=1}^{r}\Psi_{k,q}(f_i)^d\neq 0$.
    \end{proof}

\begin{theorem}[Hitting set for 
$\Sigma^{[r]}\mathord{\wedge}^{[d]}\Sigma^{[s]}\Pi^{[\delta]}$ circuits, 
when $d=\Omega(r^2)$]\label{thm:hs-depth-4}
Let $n,r,s,d,\delta$ be as in \autoref{thm:gen-4}. 
Let $\F$ be a field of characteristic $p$ such that $p=0$ or 
$p\ge rd\delta   (s^2n+\delta)$. 
Then, for every $\epsilon>0$ there is an explicit hitting set 
$\mathcal H$ of size $|\mathcal H| = O(r^4 s^4 n^2  d \delta^3 / \epsilon)$ 
such that every nonzero polynomial computed by a $\Sigma^{[r]}\mathord{\wedge}^{[d]}\Sigma^{[s]}\Pi^{[\delta]}$ circuit 
is nonzero on at least a $(1-\epsilon)$ fraction of the points in~$\mathcal H$.
\end{theorem}

\begin{proof}
By Bertrand's postulate there exists a prime 
$q\in [r^2 s^2 n + \delta + 1,\, 2r^2 s^2 n + 2\delta]$. 
Set $t = 2r^2 s^2 n + 2\delta + 1$. 
For each $k\le t$, evaluate $\Psi_{k,q}(f)$ 
on $\deg(\Psi_{k,q}(f))/\epsilon$ distinct points. 
By \autoref{thm:PIT:sparse:KS}, 
$\deg(\Psi_{k,q}(f)) \le \delta q  d 
= O(\delta d(r^2 s^2 n + \delta))$, 
so $|\mathcal H| = O( r^4 s^4 n^2 d \delta^3 / \epsilon)$.

For each good $k$, $\Psi_{k,q}(f)$ is nonzero on at least a $(1-\epsilon)$ 
fraction of its evaluation points, and at least $(1-\epsilon)t$ values of $k$ are good. 
Hence $f$ is nonzero on at least a $(1-2\epsilon)$ fraction of the points in $\mathcal H$.
\end{proof}

\begin{remark}
    Our construction of the hitting sets is not optimal and it can be  improved with more care.
\end{remark}

    % The next theorem is a combination of \autoref{thm:gen-4} and \autoref{cor:hs-depth-4} for finite fields of large enough characteristic. 

    % \begin{theorem}\label{thm:hs-depth-4-p}
    %           Let $s,r,d,n$ be integers.
    %           Let $\F$ be a field of characteristic larger than $sdn$.
    %           Let $f(\vx)\in\F[\vx]$ be a nonzero polynomial that is computable by a homogeneous $\Sigma^{[r]}\mathord{\wedge}^{[d]}\Sigma\mathord{\Pi}^{[\delta]}$ circuits with $(r-1)^2\leq d+1$, of total size $s$. Let $p$ be a prime larger than $4s^2n+1$ and let $t\geq p$.  Then, except for at most $4s(s-1)n$ values of $k\in [t]$, 
    %     $\hat{f}^{(k)}(y) \eqdef f(y,y^{(k^1 \bmod p)},\ldots, y^{(k^{n-1} \bmod p)})\neq 0$.

    %           In particular, there is an explicit hitting set of size $s^4n^4d$ for the class of size-$s$ $\Sigma^{[r]}\mathord{\wedge}^{[d]}\Sigma\Pi$ circuits, defined over $\F$.    
    % \end{theorem}

    % \begin{proof}
    %     The proof is almost identical to the proof of \autoref{thm:gen-4} and \autoref{cor:hs-depth-4} where the only difference is that we use \autoref{thm:abc-p} in place of \autoref{thm:abc}.
    % \end{proof}

\section{Reconstruction of Univariate Sums of Powers}\label{sec:univariate-rec}

%\subsection{Theory}\label{sect1}

In this section, we consider black-box reconstruction of univariate polynomials of the form  
\begin{equation}\label{eq:univariate-input}
f(x)=\sum_{i=1}^r \alpha_i (f_i(x))^d \in \F[x], \quad \deg(f_i)\coloneq \delta_i \le\delta,\; \|f_i\|_0 \le s ,
\end{equation}
where $\F$ has characteristic $0$ or characteristic $p>2rd\delta$.
As we allow coefficients $\alpha_i\in \F$, we can assume w.l.o.g.\ that each $f_i$ is a monic polynomial.

We note that we only have an upper bound on the number of summands, $r$, and in reality, $f$ can be represented as a sum of $r'<r$ terms. This distinction does not significantly affect our algorithm, as we will simply try each $r'\in[1,r]$. Thus, for simplicity, we assume that $r$ is the exact number of summands, and we continue to denote it by $r$. From now on, until we present the algorithm itself, we assume that $f$ cannot be represented as in \eqref{eq:univariate-input} with fewer terms. In particular, this implies that the polynomials $f_i^d$ are linearly independent.

As $f(x)$ is a univariate polynomial with $\deg(f)\leq d\cdot \delta$, we can interpolate its coefficients to obtain a representation 
\[
f(x)=\sum_{i=0}^{d\delta}\beta_ix^i.
\]
Hence, we may assume that this representation of $f(x)$ is given as input to the algorithm.

As described in \Cref{sec:proof-overview}, our approach is to learn a degree $r$ differential operator $L$ such that $L(f)=0$. We next describe some properties of the sought-after operator.

\subsection{The Differential Operator $L$}\label{sect:wronskian}

We shall consider the Wronskian of $f,f_1^d,\ldots,f_r^d$ (recall \Cref{sec:wronski}):
\begin{equation}\label{detW}
W(x)=W(f,f_1^d,\ldots,f_r^d)\coloneq \det\!\begin{pmatrix}
    f(x) & f_1^d(x) & \cdots & f_r^d(x)\\
    \nabla f(x) & \nabla f_1^d(x) & \cdots & \nabla f_r^d(x)\\
    \vdots &\vdots & \ddots & \vdots\\
    \nabla^{r} f(x) & \nabla^{r}f_1^d(x) & \cdots & \nabla^{r}f_r^d(x)
    \end{pmatrix}.
\end{equation}

Clearly, $W(x)\equiv0$ because the first column is a linear combination of the others. Expanding the determinant along the first column yields
\begin{equation}\label{eq:W-with-M}
\sum_{i=0}^r (-1)^i \,   W_i(x) \, \nabla^{i}f(x) \equiv 0,    
\end{equation}
where $W_i(x)$ denotes the determinant of the $r\times r$ minor obtained by deleting the $i$-th row and the first column of the matrix in \eqref{detW}. 
% For brevity, we define
% \[
% M_i(x):=(-1)^i\det W_i(x),
% \]
% so that the relation becomes
% \begin{equation}\label{eq:W-with-M}
% \sum_{i=0}^r M_i(x)\, \nabla^{i}f(x)\equiv0.    
% \end{equation}

\begin{claim}\label{Wronski}
Each polynomial $W_k(x)$, obtained from the Laplace expansion of the Wronskian, can be written as
\[
W_k(x) = \left(\Pi_{i\in[r]} f_i^{\,d-r}\right) \tilde{P}_k(x),
\]
where $\tilde{P}_k(x)$ is a polynomial of degree at most $r^2\delta$.
\end{claim}

\begin{proof}
Recall that $\deg(f_i)=\delta_i\le \delta$. Fix $k\in [[r]]$. 
For each $i\in [[r]]\setminus\{k\}$ and $j\in [r]$, we can write
\begin{equation}
    \label{eq:def-gij}
\nabla^{(i)} f_j^d \;=\; f_j^{\,d-r}\, g_{i,j}, 
\qquad \text{with}\qquad \deg(g_{i,j}) = r\delta_j - i \;\le\; r\delta \, .    
\end{equation}
Hence
\[
W_k(x)
= (-1)^{k+1} \left(\Pi_{j\in[r]} f_j^{\,d-r}\right) \det\!\big( g_{i,j}(x) \big)_{i\in [[r]]\setminus\{k\},\, j\in [r]}.
\]
Define
\begin{equation*}
    %\label{eq:P'_k}
\tilde{P}_k(x) \;:=\; (-1)^{k+1}\det\!\big( g_{i,j}(x) \big)_{i\in [[r]]\setminus\{k\},\, j\in [r]}.
\end{equation*}
Then
\[
W_k(x)=\left(\Pi_{j\in[r]} f_j^{\,d-r}\right) \tilde{P}_k(x).
\]
Since each $g_{i,j}$ has degree at most $r\delta$, $\tilde{P}_k(x)$ has degree at most $r^2\delta$, as claimed. \qedhere
\end{proof}

Let $P=\gcd(\tilde{P}_0,\ldots,\tilde{P}_r)$ and denote $P_i=\tilde{P}_i/P$. 

%\anote{We may want to mention the bit complexity as well. Perhaps we can do that when analyzing the algorithm and not here.}
\begin{corollary}\label{cor:OpL}
Assume that $p$, the characteristic of $\F$, satisfies $p=0$ or $p>d\delta$. 
Then the differential operator 
\[
\tilde{L} \;\coloneq\; \sum_{i=0}^r P_i(x)\,\nabla^{(i)} \in \F[x]\langle\nabla\rangle
\]
is not identically zero, and its solution space is spanned by $\{f_1^d,\ldots,f_r^d\}$. 
In particular, $\tilde{L}$ annihilates both $f$ and each $f_i^d$.
\end{corollary}

\begin{proof}
The fact that $\tilde{L}$ is not identically zero follows from \autoref{thm:wronskian-lin-ind}, 
since $\tilde{L}$ would remain unchanged if $f$ was replaced by any polynomial which is not a linear combination of  $\{f_1^d,\ldots,f_r^d\}$.

From \autoref{Wronski} and \eqref{eq:W-with-M}, we have
\begin{align*}
0&= \sum_{i=0}^r (-1)^i \,  W_i(x) \, \nabla^{i}f(x) \\   
&= \left(\Pi_{j\in[r]} f_j^{\,d-r}\right)\!\cdot  
  \sum_{i=0}^r \tilde{P}_i(x)\, \nabla^{i}f(x)\\
&= \left(\Pi_{j\in[r]} f_j^{\,d-r}\right)\!\cdot P \cdot 
  \sum_{i=0}^r P_i(x)\, \nabla^{i}f(x)\\
&=
\left(\Pi_{j\in[r]} f_j^{\,d-r}\right) \cdot P \cdot \tilde{L}(f)\,.
\end{align*}
Since $\left(\Pi_{j\in[r]} f_j^{\,d-r}\right)\!\cdot P$ is not identically zero, it follows that $\tilde{L}(f)\equiv 0$.
The same argument applies to any linear combination of $f_1^d,\ldots,f_r^d$, hence
\[
\mathrm{span}\{f_1^d,\ldots,f_r^d\}\subseteq\ker {\tilde{L}}.
\]

By \autoref{thm:ode}, the solution space of an operator of order $r$ has dimension at most $r$.  
As the $f_i^d$ are linearly independent, they form a basis of $\ker{ \tilde{L}}$. 
\end{proof}

%For completeness, we recall the fundamental theorem of linear homogeneous ordinary differential equations. For a proof, see, e.g., \cite[Theorem~3.10]{ODE}.

% \begin{theorem}[Fundamental theorem of linear homogeneous ordinary differential equations]\label{thm:ode}
%     The solution space of a nontrivial homogeneous differential equation of order~$k$ is $k$-dimensional. 
% \end{theorem}

We now set up a linear system whose solution  yields the desired operator $L$ from \autoref{cor:OpL}.
Consider the following system of linear equations with unknown coefficients $\gamma_{i,j}\in\F$:
\begin{equation}\label{eq:L-lin-system}
\sum_{i=0}^{r}\!\left(\sum_{j=0}^{r\delta}\gamma_{i,j}x^j\right) \nabla^{i}f(x)\;\equiv\;0\,.
\end{equation}
The system \eqref{eq:L-lin-system} consists of $\deg(f)+1$ linear equations in the unknowns $\gamma_{i,j}$, with coefficients that are linear combinations of the coefficients of~$f$.

Clearly, one solution is obtained by setting
\[
\sum_{j=0}^{r\delta}\gamma_{i,j}x^j \;=\; P_i(x)\,,
\]
which corresponds exactly to the operator $L$ in \autoref{cor:OpL}.
We next show that any nontrivial solution of the form $\sum_{j=0}^{r}Q_j\nabla^j$ satisfying $\gcd(Q_0,\ldots,Q_r)=1$ is a scalar multiple of $\tilde{L}$. In particular, any nontrivial solution that minimizes $\deg(Q_r)$ equals a scalar multiple of $\tilde{L}$.
For this, we rely on the ABC theorem for function fields discussed in \Cref{sec:abc}.

\subsubsection{Uniqueness of $L$}\label{sec:uniqueness}

%We are now ready to prove the uniqueness (up to a scalar multiple) of $\tilde{L}$ of \autoref{cor:OpL}.

\begin{claim}\label{cla:L-unique}
    Assume that the characteristic of the field $\F$ satisfies $p=0$ or $p>2rd\delta$. Let $\tilde{L}\in\F[x]\langle \nabla\rangle$ be the operator from \autoref{cor:OpL}.
    Assume $d>(r+1)^4\delta$. Then any nontrivial operator $L'=\sum_{i=0}^{r}Q_i(x)\nabla^i$ satisfying  $\deg(Q_i)\leq r^2\delta$  for all $i$, $\gcd(Q_0,\ldots,Q_r)=1$, and $L'(f)=0$ must equal a scalar multiple of $\tilde{L}$.
\end{claim}
\begin{proof}
    Let $L'$ be as in the statement of the claim. We have that 
    \begin{align*}
        0\equiv L'(f) &= \sum_{j=1}^{r}\beta_j L'(f_j^d)\\
        &=\sum_{j=1}^{r}\sum_{i=0}^{r}\beta_j Q_i(x) \nabla^{i}(f_j^d)\\
        &=\sum_{j=1}^{r} \sum_{i=0}^{r}\tilde{Q}_{i,j} f_j^{d-i}\\
        &=\sum_{j=1}^{r} \tilde{Q}_{j} f_j^{d-r}\, ,
    \end{align*}
        where  $\tilde{Q}_{i,j}$ is defined similarly to \eqref{eq:def-gij}, and $\beta_jL'(f_j^d)=\tilde{Q}_{j} f_j^{d-r}$.
        Observe that, as in \eqref{eq:def-gij}, \[\deg(\tilde{Q}_{i,j}f_j^{r-i})\leq  r^2\delta + (i\delta-i)+(r-i)\delta \,.\] Hence, $\deg(\tilde{Q}_j)\leq  (r^2+r)\delta$. Since the $f_i$ are   non-associate, 
        \[
        d-r > (r+1)^4\delta-r= (r^2+2r+1)^2\delta -r \geq (r^2+r)((r^2+r)\delta +1)\, ,
        \]
        and $p=0$ or $p\geq 2rd\delta> r((r^2+r)\delta + d\delta)$, 
        \autoref{cor:ABC} implies that $\tilde{Q}_j\equiv 0$ for all $j$. Consequently, $L'(f_j^d)=0$ for all $j$. \autoref{thm:ode} implies that the kernel of $L'$ is spanned by $f_j^d$. Observe now that the operator $Q_r\cdot\tilde{L} = P_r\cdot L'$ has degree at most $r-1$ and its kernel contains all $f_j^d$. Hence, by \autoref{thm:ode} it must be identically zero. As $\gcd(P_0,\ldots,P_r)=\gcd(Q_0,\ldots,Q_r)=1$ the claim follows. 
\end{proof}

\begin{corollary}\label{cor:unique-L}
Let $p,d$ satisfy the requirements of \autoref{cla:L-unique}. 
Then there exists a unique operator 
\[
L = \sum_{i=0}^{r} Q_i(x)\,\nabla^i
\]
satisfying $\deg Q_i \le r^2\delta$ for all $i$, 
$\gcd(Q_0,\ldots,Q_r) = 1$, $Q_r$ monic, 
and $L(f) = 0$.
Furthermore, $L$ can be computed in time $\poly(n)$.
\end{corollary}

\begin{proof}
    The claim regarding the uniqueness follows immediately from \autoref{cla:L-unique} and the fact that we chose $Q_r$ to be monic and of minimal degree.

    Since we can find, in polynomial time, a basis for the solution space of \eqref{eq:L-lin-system}, by simple diagonalization we can find a solution that minimizes $\deg(Q_r)$ such that $Q_r$ is monic. 
\end{proof}

We next show how to extract the polynomials $f_i$ given $L$. 

\subsubsection{Properties of the Solution Space of $L$}\label{sec:L-properties}

%\anote{I'm not sure that we need \autoref{cla:no-other-d-power}}

\begin{claim}\label{cla:no-other-d-power}
    Given the solution space of $L$, if $r^2+r \le d$ then there are at most $r$ polynomials in the solution space that are powers of $d$.
\end{claim}

\begin{proof}
    Assume there is a polynomial $q(x)$ such that $L(q(x)^d)=0$ and $q$ is not a scalar multiple of any of the $f_j$. Since $\{f_i^d\}_{i\in[r]}$ span the solution space of $L$, there are scalars $\alpha_i$ such that $\sum_{i}\alpha_if_i^d=q^d$. However, by choice of $d$ and $q$,  and since the polynomials $\{f_i\}\cup\{q\}$ are   non-associate, we get a contradiction to  \autoref{cor:ABC}. 
\end{proof}

Thus, our task is finding all polynomials that are $d$-th powers in $\ker{L}$. We next show that each root of each $f_j$ is also a root of $Q_r$.

\begin{definition}[$\ord_z(g)$]
    Let $g(x)$ be a polynomial. We define \emph{$\ord_z(g)$} to be the multiplicity of $z$ as a root of $g(x)$. Similarly, for an irreducible polynomial $\phi(x)$ we define $\ord_\phi(g)$ to be the largest power $e$ such that $\phi^e|g$.
\end{definition}

\begin{claim}\label{cla:qr-zero}
    If $d\geq r$ then the leading coefficient $Q_r(x)$ of $L=\sum_{i=0}^rQ_i(x)\nabla^{(i)}$ vanishes at every root $z$ of every $f_i$.
\end{claim}

\begin{proof}
Fix a zero $z\in\overline{\F}$ of some $f_j$, of multiplicity $\ge 1$.  
Then,  $\ord_z(f_j^d)=m\geq d\geq r$. It follows that for $i\leq r-1$, $\ord_z(\nabla^i(f_j^d))=m-i>m-r\geq 0$.
Assume $Q_r(z) \neq 0$.
We now have
\[0\equiv L(f_j^d) = Q_r\nabla^{r}(f_j^d)+\sum_{i=0}^{r-1}Q_{i}\nabla^{i}(f_j^d)\,.\]
Thus,
\[m-r = \ord_z (Q_r\nabla^r(f_j^d)) = \ord_z(-\sum_{i=0}^{r-1}Q_{i}\nabla^i(f_j^d))\geq m-r+1\,, \]
in contradiction. 
\end{proof}

\begin{corollary}\label{cor:irred-factor}
    Each irreducible factor, over $\F$, of each $f_i$ is an irreducible factor of $Q_r$. Furthermore, the factorization of $Q_r$ can be computed efficiently.
\end{corollary}

\begin{proof}
    This follows immediately from \Cref{cla:qr-zero} by considering minimal polynomials of roots. The claim about efficiency follows from \Cref{thm:lll} and \Cref{thm:finite-factor}, depending on the underlying field.
\end{proof}

\begin{remark}\label{rem:slower-alg}
\autoref{cor:irred-factor} tells us that, upon factorizing the leading coefficient $Q_r(x)$, we obtain a set of $m$ irreducible factors $\phi_1,\ldots,\phi_m$ of $Q_r$. 

Hence, we can iterate over all possible multiplicity vectors $\ve$, of which there are at most
${{m+\delta} \choose {\delta}} \approx (r^2 \delta)^{\delta}$, to construct candidate polynomials
\[
g_{\boldsymbol{e}}(x) = \Pi_{j \in [m]} \phi_j^{e_j}.
\]
For each such $g_{\boldsymbol{e}}$, we can apply $L$ to $g_{\boldsymbol{e}}^d$ to test whether 
$g_{\boldsymbol{e}}^d \in \ker{L}$. This yields an algorithm with running time 
$\operatorname{poly}(n, d, (r \delta)^{\delta})$.

However, our goal is to design an algorithm whose running time is polynomial in all the parameters.
\end{remark}

%\anote{Say why we can find the factors deterministically and efficiently}

\subsection{Efficiently Recovering the $f_i$}\label{sec:recover-fi}

We first note that a basis for $\ker{L}$ can be computed in polynomial time. 
Indeed, consider $d\delta+1$ unknowns $a_i$ and consider the linear system of equations
\begin{equation}
    \label{eq:L-sol-space}
    \sum_{j=0}^{r}Q_j\left(\sum_{i=0}^{\delta d}a_i x^i\right)\equiv 0 \,.
\end{equation}
Clearly, this system is solvable in time $\poly(d,\delta)$. 
Furthermore, since $\deg(f_j^d)\leq d\delta$, each $f_j^d$ lies in the solution space. 
As these polynomials span $\ker{L}$, we are guaranteed to find a basis to $\ker{L}$ in this process.
    
Let $h_1,\ldots,h_r$ be the basis found for $\ker{L}$. Thus, 
\[
\ker{L}=\Span\{f_1^d,f_2^d,\ldots,f_r^d\}
       =\Span\{h_1(x),h_2(x),...,h_r(x)\}, 
       \quad \forall i \in[r],\; \deg(h_i(x))\le\delta d \,. 
\] 

Therefore, our goal is to find all $f_j^d$ given the $h_i$. 
We achieve this by the following observation: 
if some $g\in\ker{L}$ has a root $z$ such that $\ord_z(g)\geq d$, 
where $d$ is large enough, then $z$ must be a root of one of the $f_j$. 

\begin{claim}\label{cla:high-mul-root}
    Assume $d>2r^2\delta$. 
    If $z$ is a root of order $t>r^2\delta+(r-1)$ of some polynomial 
    $h\in \Span\{f_j^d :\ j\in [r]\}$, then it is a root of one of the $f_j$.
\end{claim}

\begin{proof}
    Denote $h=\sum_{i\in[r]}c_if_i^d$. Assume, w.l.o.g., that $c_1=1$.  
    Consider the order $(r-1)$ Wronskian 
    \[
        W(f_1^d,...,f_r^d)\in\F[x],
        \quad W:=\det(\nabla^i f_j^d)_{i\in [[r-1]],j\in [r]}\,.
    \]
%    \ynote{It said $W(f_1^d,...,f_r^d)\in\C[x]^{r\times r}$ but we defined it to be a determinant not the actual matrix. Also the field is now $\F$}
    As $c_1=1$, we can replace the first column of the matrix in the Wronskian 
    with the vector whose $i$-th entry is $(\nabla^i h)$, without affecting $W$.
    Since $\ord_z(h)=t>r$, each entry satisfies $\ord_z(\nabla^i h)=t-i\geq t-(r-1)$.  
    Consequently, $z$ is a root of $W$ of order at least $t-r+1$. 
    As in the proof of  \autoref{Wronski}, the Wronskian factorizes as 
    $\Pi_{i\in[r]}f_i^{d-r+1}W'$, where $W'$ is a polynomial of degree at most $r^2\delta$. 
    Since 
    \[
        \ord_z(W)\geq t-r > r^2\delta\geq  \deg(W')\,,
    \]
    we conclude that $\ord_z(\Pi_{i\in[r]}f_i^{d-r+1})>0$. 
    Hence, $z$ is a root of one of the $f_i$.
\end{proof}

The following simple corollary is at the heart of our algorithm.

\begin{corollary}\label{cor:fi-span-high-mul-roots}
Let $d$ be as in \autoref{cla:high-mul-root}.
The subspace of all polynomials that have a root at $z$ of order at least $d$ 
is spanned by all the $f_i^d$ that have $z$ as a root.     
\end{corollary}

\begin{proof}
    Assume w.l.o.g., that among $\{f_i^d\}$ only $f_1^d,...,f_k^d$ 
    have $z$ as a root, and hence it is a root of order at least $d$ at each of them. 
    Assume for contradiction that some nonzero polynomial 
    $h\in \Span\{f_{k+1}^d,...,f_r^d\}$ satisfies $\ord_z(h)\geq d$. 
    Applying \autoref{cla:high-mul-root} to $h\in \Span\{f_{k+1}^d,...,f_r^d\}$ 
    shows that $z$ must be a root of some $f_j$ for $k<j\leq r$, a contradiction.
\end{proof}

\begin{corollary}\label{cor:fi-span-high-mul-factors}
    Let $\phi\in \F[x]$ be an irreducible polynomial. Let $d$ be as in \autoref{cla:high-mul-root}. The subspace of all polynomials that have $\phi^d$ as a factor
is spanned by all the $f_i^d$ such that $\phi$ is a factor of $f_i$. 
\end{corollary}

\begin{proof}
    Apply \autoref{cor:fi-span-high-mul-roots} to any root of $\phi$.
\end{proof}

Armed with \autoref{cor:fi-span-high-mul-factors}, 
our next step is to sieve the polynomials in $\ker{L}$ 
that have a factor of high multiplicity.

\begin{claim}\label{cla:basis-for-factor}
    Given an irreducible $\phi\in \F[x]$ such that $\deg(\phi)\leq \delta$, and $e\leq d$, we can find in time $\poly(d,r,\delta)$ 
    a basis for the space of polynomials
    \[
    \{\, g \in \Span\{h_1,\ldots,h_r\} \;\big|\; \ord_{\phi}(g)\geq e \,\}.
    \]
\end{claim}

\begin{proof}
    Let $b_1,\ldots,b_r$ and $a_0,\ldots,a_{\delta d-e\deg(\phi)}$ be indeterminates. 
    Solve the homogeneous linear system 
    \[
        \phi^e\cdot \left(\sum_{i=0}^{\delta d-e\deg(\phi)}a_ix^i\right) 
        = \sum_{i=1}^{r}b_i h_i\;.
    \]
    Any basis for the solution space immediately gives a basis for the requested subspace.
\end{proof}

As explained in \Cref{sec:proof-overview}, we recover the monic polynomials $f_i^d$ by performing a DFS (with pruning) on a suitable tree. To ease the reading, we repeat the main idea. Let $\{\phi_j\}_{j\in [m]}$ denote the irreducible factors of $Q_r$. 
As proved in \autoref{cor:irred-factor}, this set contains all factors of each $f_i$. 
Hence, each $f_i$ corresponds to an exponent vector $\ve_i$ defined by 
$\ve_{i,j} = \ord_{\phi_j}(f_i)$. 
Equivalently,
\[
    f_i(x) = \Pi_{j\in [m]} \phi_j^{e_{i,j}}.
\]
Clearly, for each $i$, we have $\sum_j e_{i,j} \deg(\phi_j) = \deg(f_i) \le \delta$.

Consider now a rooted tree of depth $\delta$, where each node has $m$ children, 
and the edge to the $j$-th child is labeled by $\phi_j$. 
Every  root-to-leaf path in the tree corresponds to an exponent vector, counting the multiplicities of the $\phi_j$ encountered along the path. 
Label each node with the exponent vector defined by the path from the root to that node, 
and associate to each such node the subspace of polynomials in $\ker{L}$ 
that have $\phi_j^d$ as a factor with multiplicity equal to the number of times 
$\phi_j$ appears along the path. 

Note, however, that this tree is highly redundant, since many distinct nodes correspond 
to the same subspace. Moreover, the tree has $m^\delta$ leaves. 

Our algorithm performs a ``smart'' depth-first search (DFS) on this tree to recover 
the $f_j$ in the lexicographic order of their exponent vectors. 
Importantly, the algorithm visits only those vertices that are guaranteed 
to lead to previously undiscovered polynomials, and thus visits in total 
only $\poly(r,\delta)$ vertices overall.

\begin{theorem}\label{thm:dfs}
Let $\{\phi_1,\ldots,\phi_m\}$ be the set of irreducible factors of $Q_r$. 
Assume $d\geq (r+1)^4\delta$ and the characteristic is $p=0$ or $p>2dr\delta$.  
Then there exists a deterministic algorithm that, given a basis 
$\{h_1,\ldots,h_r\}$ for $\ker{L}$, recovers all $f_i^d$ 
in time $\poly(m,r,\delta,d)$, or $\poly(m,r,\delta,d,\log p)$ when $p>0$.
\end{theorem}

We first give the pseudocode of the algorithm (\Cref{alg:dfs}) and then prove that it satisfies the claim in the theorem.

\begin{algorithm}[H]
\caption{\textsc{DFS-Recover-$f_i^d$}}
\label{alg:dfs}
\begin{algorithmic}[1]
\Require Basis $\mathcal{B}_{\mathbf{0}}$ of $V(\mathbf{0})=\ker{L}$; irreducible polynomials $\phi_1,\ldots,\phi_m$ of degree at most $\delta$
\Ensure A list $U$ containing all $f_i^d$
\State $U \gets \varnothing$ \Comment{set of recovered polynomials $f_i^d$}
\Procedure{DFS}{$\mathbf e,\mathcal{B}_\ve,U$} \Comment{$\ve$ is a multiplicity vector and $\cB_\ve$ spans the current ambient space}
   \If{$\dim\Span{\mathcal{B}_\ve}=1$}
       \State let $g$ be the unique monic polynomial in $\Span{\mathcal{B}_\ve}$
       \State $F(x)\gets \left(\Pi_{j=1}^m \phi_j^{e_j}\right)^d g(x)$ \label{step:alg1:add-to-list}
       \State $U\gets U\cup\{F\}$; \Return
   \EndIf
   \State $U(\mathbf e)\gets \{\,q/ (\Pi_{j=1}^m \phi_j^{e_j})^d : q\in U,\ (\Pi_{j=1}^m \phi_j^{ e_j})^d\mid q\,\}$
   \Comment{reduce $U$ to the current ambient space $V(\mathbf e)$}
   \For{$j=1$ \textbf{to} $m$}
       \State $A_j(\mathbf e)\gets\{\,g\in\Span{\mathcal{B}_\ve}:\ \ord_{\phi_j}(g)\ge d\,\}$  \label{step:alg1:aje}
       \If{$A_j(\mathbf e)\neq\{0\}$ \textbf{and} $A_j(\mathbf e)\not\subseteq\Span{U(\mathbf e)}$}\label{step:edge-condition}
           \State $\mathcal{B}_{\ve+\ve_j}\gets\{\,g/\phi_j^d: g\in$ basis$(A_j(\mathbf e))\,\}$ \label{step:alg1:B}
           \State \Call{DFS}{$\mathbf e+\mathbf e_j,\mathcal{B}_{\ve+\ve_j},U$}
       \EndIf
   \EndFor
\EndProcedure
\State \Call{DFS}{$\mathbf 0,\mathcal{B}_0,U$}
\State \Return $U$
\end{algorithmic}
\end{algorithm}

%\textbf{Invariant and correctness.}
\begin{proof}
We naturally identify nodes in the tree by their multiplicity vectors $\ve$, where the root corresponds to the all-zero vector. 

Denote by $T_\text{DFS}$ the tree defined by the DFS algorithm. That is, there is an edge between a node that is labeled by the multiplicity vector $\ve$ to its child $\ve_j$ only if the condition in Step~\ref{step:edge-condition} holds.

%Each node $u\in T_\text{DFS}$ is identified by its multiplicity vector $\mathbf e=(e_1,\ldots,e_m)\in\N^m$, where the root corresponds to the all-zero vector.
Define the vector space 
\[V(\mathbf 0)=\ker{L}\ .\]
We associate to each node $u\in T_\text{DFS}$  with multiplicity vector $\mathbf e$ two linear spaces $V(\mathbf e)$ and  $A'_j(\mathbf e)$  that are defined recursively (starting from $V(\mathbf 0)$):
\begin{equation}\label{eq:Aj}
A'_j(\mathbf e)=\{\, g\in V(\mathbf e): \ord_{\phi_j}(g)\ge d \,\} \quad \text{and}\quad 
V(\mathbf e+\mathbf e_j)=\{\, g/\phi_j^d : g\in A'_j(\mathbf e) \,\}.    
\end{equation}

We next show the relation of these spaces to the algorithm.

\begin{claim}
    For every node $u\in T_\text{DFS}$ with multiplicity vector $\mathbf e$, and $j\in[m]$ such that $\ve+\ve_j$ satisfy the condition in Step~\ref{step:edge-condition}, the following hold: 
\begin{enumerate}[label=(\Roman*)]
\item\label{item:A'j=aj} $A'_j(\mathbf e)= A_j(\mathbf e)$.
\item\label{item:V-spanned}  $V(\ve)$ is spanned by
\[
\mathcal{B}(\ve)
=\left\{
\left(\frac{f_i}{\Pi_{j=1}^m \phi_j^{e_j}}\right)^d
:\ \ord_{\phi_j}(f_i)\ge e_j\ \forall j
\right\}.
\]
\item\label{item:B=Bj} It holds that $\cB_\ve=\mathcal{B}(\ve)$.
\end{enumerate}   
\end{claim}
\begin{proof}
The proof is induction on the depth of the DFS tree, starting at the root. 

For the root,  $\mathbf e=\mathbf 0$, and $\cB(\mathbf{0})=\cB_{\mathbf{0}}=\ker{K}$. Thus,  $V(\mathbf 0)=\ker{L}=\Span{\{f_1^d,\ldots,f_r^d\}}$,
so \ref{item:V-spanned} and \ref{item:B=Bj} hold. Now, inspecting Step~\ref{step:alg1:aje} we see that \ref{item:A'j=aj} holds as well.

\noindent\emph{Induction step.}
Assume  \ref{item:A'j=aj},\ref{item:V-spanned},\ref{item:B=Bj} hold at a node $u$ with multiplicity vector $\mathbf e$, and fix $j\in[m]$ such that the node with multiplicity vector $\ve+\ve_j$ is in $T_\text{DFS}$.
By the induction hypothesis
\[
A_j(\ve)=A'_j(\ve)=\{\, g\in V(\ve):\ \ord_{\phi_j}(g)\ge d\,\}.
\]
Applying \autoref{cor:fi-span-high-mul-factors} to $V(\ve)$ with spanning set (in fact, basis) $\mathcal{B}(\ve)=\cB_\ve$ gives
\begin{align*}
    A'_j(\ve)&=\Span{\{\,h\in\mathcal{B}(\ve):\ \ord_{\phi_j}(h)\ge d\,\}}
\\ &=\Span{\left\{
\left(\tfrac{f_i}{\Pi_{t=1}^m \phi_t^{e_t}}\right)^d
:\ \ord_{\phi_j}(f_i)\ge e_j+1,\ \ord_{\phi_t}(f_i)\ge e_t\ \forall t
\right\}}.
\end{align*}
Dividing by $\phi_j^d$ yields
\[
V(\ve+\ve_j)
=\Span{\left\{
\left(\tfrac{f_i}{\Pi_{t=1}^m \phi_t^{e'_t}}\right)^d
:\ \ord_{\phi_t}(f_i)\ge e'_t\ \forall t
\right\}},
\]
where $e'_t=e_t$ for $t\neq j$ and $e'_j=e_j+1$, establishing \ref{item:V-spanned}. Moreover, since $\mathcal{B}(\ve)=\cB_\ve$ we see that $\cB_{\ve+\ve_j}$ defined in Step~\ref{step:alg1:B} satisfies $\cB_{\ve+\ve_j}=\cB(\ve+\ve_j)$ so \ref{item:B=Bj} holds. Similarly,  we see that for every $i\in[m]$ for which there is an edge from $\ve+\ve_j$ to $\ve+\ve_j+\ve_i$, the definitions of $A_i(\ve+\ve_j)$ and $A'_i(\ve+\ve_j)$
coincide.
\end{proof}

For simplicity, we slightly abuse notation and use $A_j(u)$ and $V(u)$ whenever we are at a node $u\in T_\text{DFS}$.

As a consequence of the claim we get that for every leaf $u$ (where $\dim V(u)=1$), there exists a unique $k$ such that
\[
V(u)=\Span{\left\{\left(\tfrac{f_k}{\Pi_{t=1}^m \phi_t^{e_t}}\right)^d\right\}},
\]
and for the unique monic $g\in V(u)$ we have
\[
\left(\Pi_{t=1}^m \phi_t^{ e_t}\right)^d g(x)=f_k^d(x),
\]
thus identifying the recovered polynomial, which justifies Step~\ref{step:alg1:add-to-list} of the algorithm.

At node $u$ with vector $\mathbf e$, define
\[
U(\mathbf e)
=\left\{\,q/(\Pi_{t=1}^m \phi_t^{ e_t})^d:\ q\in U,\ 
(\Pi_{t=1}^m \phi_t^{ e_t})^d\mid q\,\right\}\subseteq V(u).
\]
We recurse into child $j$ only if
\[
A_j(u)\not\subseteq\Span{U(\mathbf e)}.
\]
If $A_j(u)\subseteq\Span{U(\mathbf e)}$, then $V(u_j)\subseteq\Span{U(\mathbf e+\mathbf e_j)}$
and cannot yield a new leaf. Otherwise, the branch contains a new $f_i^d$.
Hence each $f_i^d$ is discovered exactly once, and the recursion terminates after all are found. In particular, there are  $r$ leaves.

\smallskip
\textbf{Running time.}
Each recursive descent increases $\sum_j e_j$ by one,
so at most $\sum_i \deg{f_i}\le r\delta$ nodes are visited.  
Each node requires $\poly(m,r,\delta,d)$ linear-algebra operations, 
hence the total running time is $\poly(m,r,\delta,d)$, or $\poly(m,r,\delta,d,\log p)$ when $p>0$.
\end{proof}

%\subsection{Recovering the Coefficients $\alpha_i$}\label{sec:coeff-recover}\AS{doesn't seem to require a separate section}

\subsection{The Univariate Reconstruction Algorithm}\label{sec:uni-reconstruct}

We now give the full reconstruction algorithm in the univariate case. In what follows we assume that the characteristic of $\F$ is either $p=0$ or $p>2rd\delta$. We further assume that $d>(r+1)^4\delta$.  

\begin{algorithm}[H]
\caption{\textsc{Recover Univariate Sums of Powers}}
\label{alg:recover-depth4-uni}
\begin{algorithmic}[1]
\Require Black-box access to $f(x)\in\F[x]$; parameters $r,d,\delta$.
\Ensure Representation $f(x)=\sum_{i=1}^r \alpha_i\,f_i(x)^d$ with monic $f_i$.
\State \textbf{Interpolation.} Since $\deg f\le d\delta$, query $f$ at $d\delta+1$ points and interpolate $f(x)=\sum_{i=0}^{\delta d}\beta_i x^i$.\label{alg:uni:step:interpolate}
\State \textbf{Perfect-power test.} Decide if $f=\alpha g^d$ for some scalar $\alpha$ and polynomial $g$  with $\deg(g)\le\delta$ (e.g., via repeated $\gcd(f,\nabla f)$ / squarefree decomposition). If yes, rescale $g$ to be monic and \Return $\{(\alpha_1,f_1)\}=\{(\alpha,g)\}$.\label{alg:uni:step:perfect-power}
\State \textbf{Annihilating operator (normalized).}
       Find the minimal $r'\le r$ and polynomials $Q_0,\ldots,Q_{r'}$ with $\deg Q_j\le {r'}^{2}\delta$ such that
       $L=\sum_{j=0}^{r'} Q_j(x)\,\nabla^j$ satisfies $L(f)=0$.
       Among all such solutions, choose one with minimal $\deg Q_{r'}$ and rescale so that $Q_{r'}$ is monic (as in \autoref{cor:unique-L}).\label{alg:uni:step:find-L}
\State \textbf{Factor.}
       Using \autoref{thm:lll} or \autoref{thm:finite-factor}, compute the set $\{\phi_j\}$ of all irreducible factors of degree at most $\delta$ of $Q_{r'}(x)$. By \autoref{cor:irred-factor}, $\{\phi_j\}$ contains all irreducible factors of the $f_i$’s.\label{alg:uni:step:factor}
\State \textbf{Solution space.}
       Compute a basis $\{h_1,\ldots,h_{r'}\}$ of $V=\ker{L}$ by solving \eqref{eq:L-sol-space}.\label{alg:uni:step:solution-space}
\State \textbf{Recover the $f_i^d$.}
       Run \Cref{alg:dfs} on input $(\{h_i\},\{\phi_j\})$ to obtain $U=\{f_1^d,\ldots,f_{r'}^d\}$.\label{alg:uni:step:dfs}
\State \textbf{Recover the coefficients $\alpha_i$.}
       Solve the linear system (with unknown $\alpha_i$) $f(x)=\sum_{i=1}^{r}\alpha_if_i^d$.
\State \Return $\{(\alpha_i,f_i)\}_{i=1}^{r'}$.
\end{algorithmic}
\end{algorithm}

\begin{theorem}[Correctness and bit complexity of the univariate reconstruction]\label{thm:uni-reconstruct}
Let $r,d,\delta\in \N$ with $d>(r+1)^4\delta$.
Let $\F$ be a field of characteristic $p=0$ or $p>2rd\delta$.

Suppose a nonzero univariate polynomial $f\in\F[x]$ admits a minimal representation
\[
f(x)=\sum_{i=1}^{r'} \alpha_i\,f_i(x)^d,
\qquad
\deg(f_i)\le \delta, \qquad r'\le r,
\]
where each $f_i$ is monic, and the polynomials $\{f_i\}_{i=1}^{r'}$ are   non-associate. Denote the total bit complexity of the representation with $B$.\footnote{As before, when $\F$ is a finite field, the bit complexity is $r'\delta\log|\F|$.}
Then \Cref{alg:recover-depth4-uni} outputs, in time $\poly(r,\delta,d,B)$ (or $\poly(r,\delta,d, p,\log|\F|)$ if $\F$ is a finite field of characteristic $p$),
a set $\{(\alpha'_i,f'_i)\}_{i=1}^{r'}$ such that, for some permutation $\pi$,
\[
f(x)=\sum_{i=1}^{r'} \alpha'_i\,f'_i(x)^d
\qquad\text{and}\qquad
f'_i=f_{\pi(i)},\ \ \alpha'_i=\alpha_{\pi(i)}.
\]
\end{theorem}

\begin{proof}
By assumption $\deg(f)\le d\delta$, so interpolation is efficient. Similarly, since $p=0$ or $p>2rd\delta>\deg(f)$, testing whether $f$ is a $d$th power can be done via repeated computation of $\gcd(f,\nabla f)$  (or by \autoref{thm:lll}/\autoref{thm:finite-factor}).

\autoref{cor:unique-L} guarantees that we can compute the required operator $L$ efficiently. Depending on the field, \autoref{thm:lll} and \autoref{thm:finite-factor} allow us to compute all irreducible factors of $Q_{r'}$, yielding the set $\{\phi_j\}$; by \autoref{cor:irred-factor} this set contains all irreducible factors of each $f_i$. A basis for $\ker{L}$ is obtained by solving \eqref{eq:L-sol-space}. By \autoref{thm:dfs}, we can recover all $f_i^d$ efficiently. 

Since $r^2<d$, \Cref{thm:abc}  guarantees that the polynomials $\{f_i^d\}$ are non-associate. Hence, there is a unique solution to the linear system $f(x)=\sum_{i=1}^{r}\alpha_if_i^d$. As we know $f,f_1,\ldots,f_r$ we can solve the system to recover $\alpha_1,\ldots,\alpha_r$.

%Finally, as shown in \Cref{sec:coeff-recover}, we recover the coefficients $(\alpha_i)$ by solving a linear system with $d\delta+1$ equations in $r$ variables. %\eqref{eq:mat-for-coeff}, defined using the hitting set from \autoref{thm:hs-depth-4}

Uniqueness (up to permutation) follows from \autoref{thm:abc} and the fact that the $f_i$ are monic.
\end{proof}

\section{\texorpdfstring{Reconstruction of  $\Sigma^{[r]}\wedge^{[d]}\Sigma^{[s]}\Pi^{[\delta]}$ Circuits}{Reconstruction of Sigma^[r] and^[d] Sigma^[s] Pi^[Delta] Circuits}}\label{sec:multi}

In this section we solve the reconstruction problem for $\Sigma^{[r]}\wedge^{[d]}\Sigma^{[s]}\Pi^{[\delta]}$ circuits: We are given black-box access to a polynomial
\begin{equation}\label{eq:f:multivariate2}
  f(\vx)=\sum_{i=1}^r\alpha_i f_i(\vx)^d \in \F[\vx],
\end{equation}
where each $f_i\in\F[\vx]$ is an unknown $s$-sparse polynomial of degree $\leq \delta$. Furthermore, the $f_i$ are non-associate. We also assume that $\F$ satisfies the requirements of \autoref{thm:main-rec}.

\paragraph{Overview.} The idea behind the algorithm is as follows. At a high level, we wish to reduce the problem to the univariate reconstruction problem, for which we have \Cref{alg:recover-depth4-uni}. For that, we shall restrict our input to well chosen lines, run \Cref{alg:recover-depth4-uni} for the restricted polynomials, and reconstruct the original polynomials from that information. 

In more detail, let $\cH_{\ref{thm:hs-depth-4}}$ be the hitting set given in \Cref{thm:hs-depth-4} for $\Sigma^{[2]}\mathord{\wedge}^{[d]}\Sigma^{[2s]}\Pi^{[\delta]}$ circuits, with parameter $\epsilon_{\ref{thm:hs-depth-4}} = \frac{1}{2r^2}$. 
Let $\cH_{\ref{construction:robust-ks}}$ be the robust interpolating set for $n$-variate, $s$-sparse polynomials of degree $\delta$, given  in \Cref{construction:robust-ks}, for parameter $\epsilon_{\ref{construction:robust-ks}} = \frac{1}{100 n r^2}$.

For every $(\vu,\vv)\in \cH_{\ref{thm:hs-depth-4}}\times \cH_{\ref{construction:robust-ks}}$ we define the univariate polynomial
\[F_{\vu,\vv}(t)=f(\vu+t(\vv-\vu)) = \sum_{i=1}^{r}\alpha_i f_i(\vu+t(\vv-\vu))^d\ .\]
We will show that for most points $(\vu,\vv)\in \cH_{\ref{thm:hs-depth-4}}\times \cH_{\ref{construction:robust-ks}}$, $F_{\vu,\vv}$ cannot be represented as a sum of fewer than $r$ powers. Therefore, \Cref{alg:recover-depth4-uni} will return for each ``good'' $(\vu,\vv)$ a list of $r$ pairs $(\lambda_{i,\vu,\vv},h_{i,\vu,\vv}(t))$ such that, up to reordering, $\lambda_{i,\vu,\vv}h_{i,\vu,\vv}(t)^d=\alpha_if_i(\vu+t(\vv-\vu))^d$. In particular, $\lambda_{i,\vu,\vv}h_{i,\vu,\vv}(1)^d= \alpha_if_i(\vv)^d$. Thus, we have access to evaluations of $f_i$ on enough points of $\cH_{\ref{construction:robust-ks}}$. This will enable us to interpolate $f_i$ and conclude the algorithm. An important point is that we will have to ``align'' the outputs of the algorithm when run on different $\vv\in \cH_{\ref{construction:robust-ks}}$ so that we know to map each output to a unique $f_i$. This will be done by noting that with high probability over the choice of $u$, the \emph{labels} $\alpha_if_i(\vu)^d=\lambda_{i,\vu,\vv}h_{i,\vu,\vv}(0)^d$ are distinct.

We next give the description of the algorithm and then its formal analysis.

\begin{algorithm}[H]
\caption{Multivariate $\Sigma^{[r]}\wedge^{[d]}\Sigma^{[s]}\Pi^{[\delta]}$ Reconstruction}
\label{alg:multivariate-reconstruction}
\begin{algorithmic}[1]
\Require Sets $\cH_{\ref{thm:hs-depth-4}}$, $\cH_{\ref{construction:robust-ks}}$; blackbox access to $f=\sum_{i=1}^{r}\alpha_if_i^d$; parameters $r,\epsilon,d,\delta$.
\Ensure Output   $\left\{(\lambda_{i },h_{i })\right\}$, such that $\sum_{i=1}^{r}\lambda_{i } h_{i }(\vx)^d=f(\vx)$.

\ForAll{$(\vu,\vv)\in \cH_{\ref{thm:hs-depth-4}}\times \cH_{\ref{construction:robust-ks}}$} \Comment{Run univariate recovery for each $(\vu,\vv)$}
    \State Let \[F_{\vu,\vv}(t):=f(u+t (\vv-\vu))\ , \quad \text{and} \quad f_{i,\vu,\vv}(t):=f_i(\vu+t(\vv-\vu))\] 
    \State Run \Cref{alg:recover-depth4-uni} on $F_{\vu,\vv}$
    \If{the algorithm did not abort}
        \State Let $\{(\lambda_{i,\vu,\vv},h_{i,\vu,\vv}(t))\}_{i=1}^{r_{\vu,\vv}}$ be its output
    \EndIf
\EndFor

\State $r \gets \max_{\vu,\vv} r_{\vu,\vv}$

\ForAll{$\vu\in \cH_{\ref{thm:hs-depth-4}}$} \Comment{Define good $\vv$'s per $\vu$}

    \State $\cV_\vu \gets \{\, \vv\in \cH_{\ref{construction:robust-ks}} \;:\; r=r_{\vu,\vv} \text{ and the } r \text{ labels } 
    \lambda_{i,\vu,\vv}h_{i,\vu,\vv}(0)^d \text{ are nonzero and distinct } \,\}$
\EndFor

\State $\cU \gets \{\, \vu : |\cV_\vu|\geq (1-\epsilon\binom{r}{2})|\cH_\text{KS}| \,\}$  \Comment{Define good $\vu$'s}

\State Fix some $\vu\in \cU$
\ForAll{$\vv\in \cV_\vu$} \Comment{Align indices across different $\vv$'s} \label{step:alignment}
    \State Reorder $\{(\lambda_{i,\vu,\vv},h_{i,\vu,\vv}(t))\}_{i=1}^{r}$ so that for every $\vv\neq \vv'\in \cV_\vu$,
    \[
      \lambda_{i,\vu,\vv}h_{i,\vu,\vv}(0)^d=\lambda_{i,\vu,\vv'}h_{i,\vu,\vv'}(0)^d \ ,
    \]
    and denote $\lambda_{i,\vu}:=\lambda_{i,\vu,\vv}h_{i,\vu,\vv}(0)^d$
\EndFor
\ForAll{$\vv\in \cV_\vu$} \Comment{Reconstruct the $g_i$}
    \For{$i=1$ to $r$}
        \State $h_{i,\vu}(\vv) \gets  \frac{h_{i,\vu,\vv}(1)}{h_{i,\vu,\vv}(0)} $ \label{step:g}
    \EndFor
\EndFor
\State Apply \Cref{thm:robust-ks} to reconstruct each $h_{i,\vu}$ as a degree $\delta$ $s$-sparse polynomial \label{step:interpolate-hiu}

\State \Return the set of pairs $\left\{(\lambda_{i,\vu},h_{i,\vu})\right\}$
\end{algorithmic}
\end{algorithm}

\begin{theorem}\label{cla:multivariate-reconstruction-alg}
    Let $r,d,\delta\in \N$ with $d>(r+1)^4\delta$. Let $\F$ be a field of characteristic $p=0$ or $p\ge rd\delta   (s^2n+\delta)$. Then,
    \Cref{alg:multivariate-reconstruction} runs in time $\poly(n,s,d)$ and its output satisfies $\sum_{i=1}^{r}\lambda_{i,\vu} h_{i,\vu}(\vx)^d=f(\vx)$.    
\end{theorem}

\begin{proof}
    As the sets involved have polynomial size and each step requires polynomial time, the claim regarding the running time is clear. We next prove correctness.

    The proof is composed of several claims. We shall first prove that $\cU$ is non-empty and that for every $\vu\in \cU$, $|\cV_\vu|> (1-\frac{1}{100n})|\cH_{\ref{construction:robust-ks}}|$. We shall then show that the alignment process succeeds and that it gives access to evaluations of the different $f_i$. We then claim (using \Cref{thm:robust-ks}) that since $\cV_\vu$ is large and the alignment process succeeded, the $g_i$ that we reconstructed are scalar multiples of the $f_i$. 

    \begin{claim}[Many good $\vu$'s]\label{cla:u}
Let $\cH_{\ref{thm:hs-depth-4}}$ be the hitting set as above.
Then, all but at most $\epsilon_{\ref{thm:hs-depth-4}}\left(r+\binom r2\right)\,|\cH_{\ref{thm:hs-depth-4}}|<\frac{1}{2}|\cH_{\ref{thm:hs-depth-4}}|$
points $\vu\in\cH_{\ref{thm:hs-depth-4}}$ satisfy the following simultaneously:
\begin{enumerate}
\item $f_i(\vu)\neq 0$ for all $i\in[r]$, and
\item the values $\Lambda_i:=\alpha_i f_i(\vu)^d$ are pairwise distinct.
\end{enumerate}
We call the points $\vu\in \cH_{\ref{thm:hs-depth-4}}$ that satisfy both requirements \emph{good}.
\end{claim}

\begin{proof}
Since each $f_i$ is $s$-sparse, we get from \Cref{thm:hs-depth-4} and
the union bound that there are at most $\epsilon_{\ref{thm:hs-depth-4}} r |\cH_{\ref{thm:hs-depth-4}}|$, on which any of the $f_i$ vanishes.
Next, consider the $\binom{r}{2}$ polynomials \[
P_{i,j}(\vx):=\alpha_i f_i(\vx)^d-\alpha_j f_j(\vx)^d.  
\]
Clearly, $P_{i,j}$ is a $\Sigma^{[2]}\mathord{\wedge}^{[d]}\Sigma^{[2s]}\Pi^{[\delta]}$ circuit.
\Cref{thm:hs-depth-4} and the union bound imply that there are at most $\epsilon_{\ref{thm:hs-depth-4}} \binom{r}{2} |\cH_{\ref{thm:hs-depth-4}}|$, on which any of the $P_{i,j}$ vanishes. Combining the two estimates, with our choice of $\epsilon_{\ref{thm:hs-depth-4}}=1/2r^2$, we conclude the proof.
\end{proof}
    We next prove that for any such good $\vu$, the set $\cV_\vu$ is large.

\begin{claim}[$\cV_\vu$ is large for a good $\vu$]\label{cla:good-direction-from-H}
%Let $\cH_{\ref{thm:hs-depth-4}}$ as in \Cref{cla:u}. Let $\cH_{\ref{construction:robust-ks}}$ be the interpolating set for $s$-sparse polynomials guaranteed in \Cref{lem:robust-KS}.
Let $\vu\in\cH_{\ref{thm:hs-depth-4}}$ be a good point.  
It holds that $|\cV_\vu|\geq (1- \epsilon_{\ref{construction:robust-ks}}\binom r2)|\cH_{\ref{construction:robust-ks}}|$ and for every $\vv\in\cV_\vu$ the following sets of labels are equal
\[
\left\{\lambda_{i,\vu,\vv} h_{i,\vu,\vv}(0)^d \mid i\in[r]\right\} = \left\{ \alpha_if_i(\vu)^d \mid i\in[r] \right\} \ .
\]
\end{claim}

\begin{proof}
For $1\le i<j\le r$ denote
\[
g_{i,j}(\vx)\ :=\ f_j(\vu)\,f_i(\vx)\;-\;f_i(\vu)\,f_j(\vx)\ \in\ \F[\vx].
\]
Since $\vu$ is good we have that $f_i(\vu),f_j(\vu)\neq 0$. As the $f_i$ are non-associate, we conclude that $g_{i,j}\neq  0$.
From  \autoref{cor:ks-lin-ind} and the union bound we conclude that all but at most $\epsilon_{\ref{construction:robust-ks}}\binom r2\,|\cH_{\ref{construction:robust-ks}}|$ points $\vv\in \cH_{\ref{construction:robust-ks}}$ satisfy $g_{i,j}(\vv)\neq 0$ for every $i<j$.
Let $\vv$ be such a point (on which no $g_{i,j}$ vanishes). We next show that $\vv\in\cV_\vu$.  

We first note that the univariate polynomials \[f_{i,\vu,\vv}:=f_i(\vu+t(\vv-\vu))\] are non-associate. Indeed, if $f_{i,\vu,\vv} \sim f_{j,u,\vv}$ then this would give 
\[\frac{f_i(\vv)}{f_i(\vu)} = \frac{f_{i,\vu,\vv}(1)}{f_{i,\vu,\vv}(0)} = \frac{f_{j,u,\vv}(1)}{f_{j,u,\vv}(0)} =\frac{f_j(\vv)}{f_j(\vu)}\ ,\] 
implying $g_{i,j}(\vv)=0$, in contradiction.
Consequently the representation
\[
F_{\vu,\vv}(t)=\sum_{i=1}^r \alpha_i\,f_{i,\vu,\vv}(t)^d
\]
has minimal top fan-in exactly $r$, so \Cref{alg:recover-depth4-uni} returns exactly $r$ polynomials. Since the representation is unique we have that up to reordering, $\lambda_{i,\vu,\vv} h_{i,\vu,\vv}(t)^d  =  \alpha_if_{i,\vu,\vv}(t)^d$. Consequently,  
\[\lambda_{i,\vu,\vv} h_{i,\vu,\vv}(0)^d  =  \alpha_if_{i,\vu,\vv}(0)^d = \alpha_if_i(\vu)^d\ . \]
This concludes the proof regarding the size of $\cV_\vu$.

Now, observe that for any $\vv\in\cV_\vu$, since \Cref{alg:recover-depth4-uni} returned exactly $r$ polynomials, it must be the case that the $f_{i,\vu,\vv}$ are linearly independent, and the same argument implies that $\left\{\lambda_{i,\vu,\vv} h_{i,\vu,\vv}(0)^d \mid i\in[r]\right\} = \left\{ \alpha_if_i(\vu)^d \mid i\in[r] \right\}$.
\end{proof}

Combining \Cref{cla:u} and \Cref{cla:good-direction-from-H}, we conclude that the set $\cU$ is not empty (indeed, it is quite large) and that for every $\vu\in \cU$, for each $\vv\in\cV_\vu$ the set of labels equals $\left\{ \alpha_if_i(\vu)^d \mid i\in[r] \right\}$. In particular, we can align the labels $\lambda_{i,\vu,\vv} h_{i,\vu,\vv}(0)^d $ across different $v\in\cV_\vu$. Thus, we can assume without loss of generality that for every  $v\in \cV_\vu$ it holds that 
\[\lambda_{i,\vu,\vv} h_{i,\vu,\vv}(0)^d  =  \alpha_if_{i,\vu,\vv}(0)^d = \alpha_if_i(\vu)^d\ . \]
Hence, by uniqueness of representation as small sum of $d$-th powers (\Cref{thm:abc}), we have that
for all $\vv\in \cV_\vu$ it holds that 
\[\lambda_{i,\vu,\vv} h_{i,\vu,\vv}(t)^d  =  
\alpha_if_{i,\vu,\vv}(t)^d\ . \]
Recall that \Cref{alg:recover-depth4-uni} returns monic polynomials, so each $h_{i,\vu,\vv}$ is monic. As the leading coefficient of $t$ in $f_{i,\vu,\vv}(t)$ equals $f_{i}(\vv-\vu)$ we have that $\lambda_{i,\vu,\vv}=\alpha_if_i(\vv-\vu)^d$. From this we obtain that $\lambda_{i,\vu,\vv}=\lambda_{i,\vu,\vv'}$ for all $\vv,\vv'\in\cV_\vu$. We can thus rename $\lambda_{i,\vu,\vv}$ to $\lambda_{i,\vu}$.
Note that the polynomial
$\tilde{h}_{i,\vu,\vv}(t):=h_{i,\vu,\vv}(t)/h_{i,\vu,\vv}(0)$ satisfies $\tilde{h}_{i,\vu,\vv}(0)=1$. Moreover,
\[\tilde{h}_{i,\vu,\vv}(t)^d = \left(\frac{h_{i,\vu,\vv}(t)}{h_{i,\vu,\vv}(0)}\right)^d= \left(\frac{f_{i,\vu,\vv}(t)}{f_i(\vu)}\right)^d  \]
Since $f_{i,\vu,\vv}({0})/f_i(\vu)=1$, we obtain that $\tilde{h}_{i,\vu,\vv}(t)=f_{i,\vu,\vv}(t)/f_i(\vu)$. I.e., it is the unique $d$-th root of the polynomial $\left(\frac{f_{i,\vu,\vv}(t)}{f_i(\vu)}\right)^d$ whose value at $0$ equals $1$ (other roots evaluate to different roots of unity).
In particular,  $\tilde{h}_{i,\vu,\vv}(1)=\frac{f_{i,\vu,\vv}(1)}{f_i(\vu)}=\frac{f_i(\vv)}{f_i(\vu)}$. 
Thus, the function $h_{i,\vu}$ defined in Step~\ref{step:g} satisfies that for all $\vv\in\cV_\vu$, 
\[h_{i,\vu}(\vv)=\frac{h_{i,\vu,\vv}(1)}{h_{i,\vu,\vv}(0)}=\tilde{h}_{i,\vu,\vv}(1)=f_i(\vv)/f_i(\vu) \ .\]
Since $f_i$ is an $s$-sparse polynomial of degree $\delta$, and we have its evaluations on all $\vv\in \cV_\vu$ which consists of at least $|\cV_\vu|\geq (1- \epsilon_{\ref{construction:robust-ks}}\binom r2)|\cH_{\ref{construction:robust-ks}}|\geq(1-\frac{1}{100n})|\cH_{\ref{construction:robust-ks}}|$ points of $\cH_{\ref{construction:robust-ks}}$ we conclude by \Cref{thm:robust-ks} that Step~\ref{step:interpolate-hiu} returns the polynomial $f_i(\vv)/f_i(\vu)$.

Observe that if $\beta_i$ satisfies
\[\beta_i \left( \frac{f_i(\vv)}{f_i(\vu)}\right)^d = \alpha_i f_i(\vv)^d \ ,\]
then 
\[\beta_i = \alpha_i f_i(\vu)^d=\lambda_{i,\vu,\vv}h_{i,\vu,\vv}(0)^d = \lambda_{i,\vu}\,.\]
Thus, the algorithm indeed returns $(\lambda_{i,\vu},h_{i,\vu})$ satisfying
\[\lambda_{i,\vu}h_{i,\vu}(\vx)^d=\alpha_i f_i(\vx)^d \ .\]
This concludes the analysis of \Cref{alg:multivariate-reconstruction}.
\end{proof}

\bibliographystyle{alpha}
\bibliography{sources}

@InProceedings{peleg2022tensor,
  author =	{Peleg, Shir and Shpilka, Amir and Volk, Ben Lee},
  title =	{{Tensor Reconstruction Beyond Constant Rank}},
  booktitle =	{15th Innovations in Theoretical Computer Science Conference (ITCS 2024)},
  pages =	{87:1--87:20},
  series =	{Leibniz International Proceedings in Informatics (LIPIcs)},
  ISBN =	{978-3-95977-309-6},
  ISSN =	{1868-8969},
  year =	{2024},
  volume =	{287},
  editor =	{Guruswami, Venkatesan},
  publisher =	{Schloss Dagstuhl -- Leibniz-Zentrum f{\"u}r Informatik},
  address =	{Dagstuhl, Germany},
  URL =		{https://drops.dagstuhl.de/entities/document/10.4230/LIPIcs.ITCS.2024.87},
  URN =		{urn:nbn:de:0030-drops-196157},
  doi =		{10.4230/LIPIcs.ITCS.2024.87}
}

@inproceedings{Kay12,
author = {Kayal, Neeraj},
title = {Affine projections of polynomials: extended abstract},
year = {2012},
isbn = {9781450312455},
publisher = {Association for Computing Machinery},
address = {New York, NY, USA},
url = {https://doi.org/10.1145/2213977.2214036},
doi = {10.1145/2213977.2214036},
booktitle = {Proceedings of the Forty-Fourth Annual ACM Symposium on Theory of Computing},
pages = {643–662},
numpages = {20},
location = {New York, New York, USA},
series = {STOC '12}
}

@inproceedings{KS01,
author = {Klivans, Adam R. and Spielman, Daniel},
title = {Randomness efficient identity testing of multivariate polynomials},
year = {2001},
isbn = {1581133499},
publisher = {Association for Computing Machinery},
address = {New York, NY, USA},
url = {https://doi.org/10.1145/380752.380801},
doi = {10.1145/380752.380801},
booktitle = {Proceedings of the Thirty-Third Annual ACM Symposium on Theory of Computing},
pages = {216–223},
numpages = {8},
location = {Hersonissos, Greece},
series = {STOC '01}
}

@INPROCEEDINGS{GKS20,
  author={Garg, Ankit and Kayal, Neeraj and Saha, Chandan},
  booktitle={2020 IEEE 61st Annual Symposium on Foundations of Computer Science (FOCS)}, 
  title={Learning sums of powers of low-degree polynomials in the non-degenerate case}, 
  year={2020},
  volume={},
  number={},
  pages={889-899},
  doi={10.1109/FOCS46700.2020.00087}}

@book{ODE,
  author    = {Gerald Teschl},
  title     = {Ordinary Differential Equations and Dynamical Systems},
  series    = {Graduate Studies in Mathematics},
  volume    = {140},
  publisher = {American Mathematical Society},
  address   = {Providence, RI},
  year      = {2012},
  isbn      = {978-0-8218-8328-0},
}

@article{LLL,
author = {Lenstra, H.W. jr. and Lenstra, A.K. and Lovász, L.},
journal = {Mathematische Annalen},
pages = {515-534},
title = {Factoring Polynomials with Rational Coefficients.},
url = {http://eudml.org/doc/182903},
volume = {261},
year = {1982},
}

@book{mason, 
place={Cambridge}, 
series={London Mathematical Society Lecture Note Series}, 
title={Diophantine Equations over Function Fields}, DOI={10.1017/CBO9780511752490}, publisher={Cambridge University Press}, author={Mason, Richard C.}, year={1984}, collection={London Mathematical Society Lecture Note Series}}

@article{stothers1981polynomial,
  title={{Polynomial identities and Hauptmoduln}},
  author={Stothers, Walter W.},
  journal={The Quarterly Journal of Mathematics},
  volume={32},
  number={3},
  pages={349--370},
  year={1981},
  publisher={Oxford University Press}
}

@article{vaserstein2003vanishing,
	title={Vanishing polynomial sums},
	author={Vaserstein, Leonid N and Wheland, Ethel R},
	journal={Communications in Algebra},
	volume={31},
	number={2},
	pages={751--772},
	year={2003},
	publisher={Taylor \& Francis}
}

@inproceedings{BhargavaSV21,
  author       = {Vishwas Bhargava and
                  Shubhangi Saraf and
                  Ilya Volkovich},
  editor       = {Samir Khuller and
                  Virginia Vassilevska Williams},
  title        = {Reconstruction algorithms for low-rank tensors and depth-3 multilinear
                  circuits},
  booktitle    = {{STOC} '21: 53rd Annual {ACM} {SIGACT} Symposium on Theory of Computing,
                  Virtual Event, Italy, June 21-25, 2021},
  pages        = {809--822},
  publisher    = {{ACM}},
  year         = {2021},
  url          = {https://doi.org/10.1145/3406325.3451096},
  doi          = {10.1145/3406325.3451096},
  bibsource    = {dblp computer science bibliography, https://dblp.org}
}

@inproceedings{BhargavaS25,
  author       = {Vishwas Bhargava and
                  Devansh Shringi},
  editor       = {Keren Censor{-}Hillel and
                  Fabrizio Grandoni and
                  Jo{\"{e}}l Ouaknine and
                  Gabriele Puppis},
  title        = {Faster {\&} Deterministic {FPT} Algorithm for Worst-Case Tensor
                  Decomposition},
  booktitle    = {52nd International Colloquium on Automata, Languages, and Programming,
                  {ICALP} 2025, July 8-11, 2025, Aarhus, Denmark},
  series       = {LIPIcs},
  volume       = {334},
  pages        = {28:1--28:20},
  publisher    = {Schloss Dagstuhl - Leibniz-Zentrum f{\"{u}}r Informatik},
  year         = {2025},
  url          = {https://doi.org/10.4230/LIPIcs.ICALP.2025.28},
  doi          = {10.4230/LIPICS.ICALP.2025.28},
  bibsource    = {dblp computer science bibliography, https://dblp.org}
}

@article {GathenShoup1992-deterministic-factor,
    AUTHOR = {von zur Gathen, Joachim and Shoup, Victor},
     TITLE = {Computing {F}robenius maps and factoring polynomials},
   JOURNAL = {Comput. Complexity},
  FJOURNAL = {Computational Complexity},
    VOLUME = {2},
      YEAR = {1992},
    NUMBER = {3},
     PAGES = {187--224},
      ISSN = {1016-3328,1420-8954},
   MRCLASS = {12Y05 (11T06 11Y16 68Q40)},
  MRNUMBER = {1220071},
MRREVIEWER = {Heinrich\ Rolletschek},
       DOI = {10.1007/BF01272074},
       URL = {https://doi.org/10.1007/BF01272074},
}

@article{KayalSaxena07,
	Author = {Neeraj Kayal and Nitin Saxena},
	Bibsource = {dblp computer science bibliography, http://dblp.org},
	Doi = {10.1007/s00037-007-0226-9},
	Journal = {Computational Complexity},
	Number = {2},
	Pages = {115--138},
	Title = {Polynomial Identity Testing for Depth 3 Circuits},
	Url = {http://dx.doi.org/10.1007/s00037-007-0226-9},
	Volume = {16},
	Year = {2007}}

@inproceedings{KayalSaraf09,
	Author = {Neeraj Kayal and Shubhangi Saraf},
	Bibsource = {dblp computer science bibliography, http://dblp.org},
	Booktitle = {50th Annual {IEEE} Symposium on Foundations of Computer Science, {FOCS} 2009, October 25-27, 2009, Atlanta, Georgia, {USA}},
	Crossref = {DBLP:conf/focs/2009},
	Doi = {10.1109/FOCS.2009.67},
	Pages = {198--207},
	Title = {Blackbox Polynomial Identity Testing for Depth 3 Circuits},
	Url = {http://dx.doi.org/10.1109/FOCS.2009.67},
	Year = {2009}}

@proceedings{DBLP:conf/focs/2009,
	Bibsource = {dblp computer science bibliography, http://dblp.org},
	Isbn = {978-0-7695-3850-1},
	Publisher = {{IEEE} Computer Society},
	Title = {50th Annual {IEEE} Symposium on Foundations of Computer Science, {FOCS} 2009, October 25-27, 2009, Atlanta, Georgia, {USA}},
	Url = {http://ieeexplore.ieee.org/xpl/mostRecentIssue.jsp?punumber=5438528},
	Year = {2009}}

@inproceedings{KarninS09rec,
	Author = {Zohar S. Karnin and Amir Shpilka},
	Bibsource = {dblp computer science bibliography, http://dblp.org},
	Booktitle = {Proceedings of the 24th Annual {IEEE} Conference on Computational Complexity, {CCC} 2009, Paris, France, 15-18 July 2009},
	Crossref = {DBLP:conf/coco/2009},
	Doi = {10.1109/CCC.2009.18},
	Pages = {274--285},
	Title = {Reconstruction of Generalized Depth-3 Arithmetic Circuits with Bounded Top Fan-in},
	Url = {http://doi.ieeecomputersociety.org/10.1109/CCC.2009.18},
	Year = {2009}}

@proceedings{DBLP:conf/coco/2009,
	Bibsource = {dblp computer science bibliography, http://dblp.org},
	Isbn = {978-0-7695-3717-7},
	Publisher = {{IEEE} Computer Society},
	Title = {Proceedings of the 24th Annual {IEEE} Conference on Computational Complexity, {CCC} 2009, Paris, France, 15-18 July 2009},
	Year = {2009}}

@article{KarninMSV13,
	Author = {Zohar S. Karnin and Partha Mukhopadhyay and Amir Shpilka and Ilya Volkovich},
	Bibsource = {dblp computer science bibliography, http://dblp.org},
	Doi = {10.1137/110824516},
	Journal = {{SIAM} J. Comput.},
	Number = {6},
	Pages = {2114--2131},
	Title = {Deterministic Identity Testing of Depth-4 Multilinear Circuits with Bounded Top Fan-in},
	Url = {http://dx.doi.org/10.1137/110824516},
	Volume = {42},
	Year = {2013}}

@inproceedings{Saxena08diagonal,
	Author = {Nitin Saxena},
	Bibsource = {dblp computer science bibliography, http://dblp.org},
	Booktitle = {Automata, Languages and Programming, 35th International Colloquium, {ICALP} 2008, Reykjavik, Iceland, July 7-11, 2008, Proceedings, Part {I:} Tack {A:} Algorithms, Automata, Complexity, and Games},
	Crossref = {DBLP:conf/icalp/2008-1},
	Doi = {10.1007/978-3-540-70575-8_6},
	Pages = {60--71},
	Title = {Diagonal Circuit Identity Testing and Lower Bounds},
	Url = {http://dx.doi.org/10.1007/978-3-540-70575-8_6},
	Year = {2008}
}

@proceedings{DBLP:conf/icalp/2008-1,
	Bibsource = {dblp computer science bibliography, http://dblp.org},
	Editor = {Luca Aceto and Ivan Damg{\aa}rd and Leslie Ann Goldberg and Magn{\'{u}}s M. Halld{\'{o}}rsson and Anna Ing{\'{o}}lfsd{\'{o}}ttir and Igor Walukiewicz},
	Isbn = {978-3-540-70574-1},
	Publisher = {Springer},
	Series = {Lecture Notes in Computer Science},
	Title = {Automata, Languages and Programming, 35th International Colloquium, {ICALP} 2008, Reykjavik, Iceland, July 7-11, 2008, Proceedings, Part {I:} Tack {A:} Algorithms, Automata, Complexity, and Games},
	Volume = {5125},
	Year = {2008}}

@article{KarninS11pit,
	Author = {Zohar S. Karnin and Amir Shpilka},
	Bibsource = {dblp computer science bibliography, http://dblp.org},
	Doi = {10.1007/s00493-011-2537-3},
	Journal = {Combinatorica},
	Number = {3},
	Pages = {333--364},
	Title = {Black box polynomial identity testing of generalized depth-3 arithmetic circuits with bounded top fan-in},
	Url = {http://dx.doi.org/10.1007/s00493-011-2537-3},
	Volume = {31},
	Year = {2011}}

@article{SaxenaSesh12pit,
	Author = {Nitin Saxena and Comandur Seshadhri},
	Bibsource = {dblp computer science bibliography, http://dblp.org},
	Doi = {10.1137/10848232},
	Journal = {{SIAM} J. Comput.},
	Number = {5},
	Pages = {1285--1298},
	Title = {Blackbox Identity Testing for Bounded Top-Fanin Depth-3 Circuits: The Field Doesn't Matter},
	Url = {http://dx.doi.org/10.1137/10848232},
	Volume = {41},
	Year = {2012}}

@article{SaxenaSesh13rank,
	Author = {Nitin Saxena and Comandur Seshadhri},
	Bibsource = {dblp computer science bibliography, http://dblp.org},
	Doi = {10.1145/2528403},
	Journal = {J. {ACM}},
	Number = {5},
	Pages = {33},
	Title = {From Sylvester-Gallai configurations to rank bounds: Improved blackbox identity test for depth-3 circuits},
	Url = {http://doi.acm.org/10.1145/2528403},
	Volume = {60},
	Year = {2013}}

@article{RazShpilka05,
	Author = {Ran Raz and Amir Shpilka},
	Bibsource = {dblp computer science bibliography, http://dblp.org},
	Doi = {10.1007/s00037-005-0188-8},
	Journal = {Computational Complexity},
	Number = {1},
	Pages = {1--19},
	Title = {Deterministic polynomial identity testing in non-commutative models},
	Url = {http://dx.doi.org/10.1007/s00037-005-0188-8},
	Volume = {14},
	Year = {2005}}

@article{SarafV18,
  author       = {Shubhangi Saraf and
                  Ilya Volkovich},
  title        = {Black-Box Identity Testing of Depth-4 Multilinear Circuits},
  journal      = {Comb.},
  volume       = {38},
  number       = {5},
  pages        = {1205--1238},
  year         = {2018},
  url          = {https://doi.org/10.1007/s00493-016-3460-4},
  doi          = {10.1007/S00493-016-3460-4},
  bibsource    = {dblp computer science bibliography, https://dblp.org}
}

@inproceedings{Zippel79,
	Author = {Richard Zippel},
	Bibsource = {dblp computer science bibliography, http://dblp.org},
	Booktitle = {Symbolic and Algebraic Computation, {EUROSAM} '79, An International Symposiumon Symbolic and Algebraic Computation, Marseille, France, June 1979, Proceedings},
	Crossref = {DBLP:conf/eurosam/1979},
	Doi = {10.1007/3-540-09519-5_73},
	Pages = {216--226},
	Title = {Probabilistic algorithms for sparse polynomials},
	Url = {http://dx.doi.org/10.1007/3-540-09519-5_73},
	Year = {1979}}

@inproceedings{forbes2014hitting,
  title={Hitting sets for multilinear read-once algebraic branching programs, in any order},
  author={Forbes, Michael A and Saptharishi, Ramprasad and Shpilka, Amir},
  booktitle={Proceedings of the forty-sixth annual ACM symposium on Theory of computing},
  pages={867--875},
  year={2014}
}

@article{Shpilka09-rec-sps,
	Author = {Amir Shpilka},
	Bibsource = {dblp computer science bibliography, http://dblp.org},
	Doi = {10.1137/070694879},
	Journal = {{SIAM} J. Comput.},
	Number = {6},
	Pages = {2130--2161},
	Title = {Interpolation of Depth-3 Arithmetic Circuits with Two Multiplication Gates},
	Url = {http://dx.doi.org/10.1137/070694879},
	Volume = {38},
	Year = {2009}}

@article{Schwartz80,
	Author = {Jacob T. Schwartz},
	Bibsource = {dblp computer science bibliography, http://dblp.org},
	Doi = {10.1145/322217.322225},
	Journal = {J. {ACM}},
	Number = {4},
	Pages = {701--717},
	Title = {Fast Probabilistic Algorithms for Verification of Polynomial Identities},
	Url = {http://doi.acm.org/10.1145/322217.322225},
	Volume = {27},
	Year = {1980}}

@article{DemilloL78,
	Author = {Richard A. DeMillo and Richard J. Lipton},
	Bibsource = {dblp computer science bibliography, http://dblp.org},
	Doi = {10.1016/0020-0190(78)90067-4},
	Journal = {Inf. Process. Lett.},
	Number = {4},
	Pages = {193--195},
	Title = {A Probabilistic Remark on Algebraic Program Testing},
	Url = {http://dx.doi.org/10.1016/0020-0190(78)90067-4},
	Volume = {7},
	Year = {1978}}

@article{Ore1922,
  author    = {{\O}ystein Ore},
  title     = {{\"U}ber h{\"o}here Kongruenzen},
  journal   = {Norske Videnskaps-Akademi i Oslo. Forhandlinger (Proceedings of the Norwegian Academy of Science and Letters)},
  year      = {1922},
  volume    = {1922},
  number    = {12},
  pages     = {1--8},
  language  = {German},
  note      = {Contains the original root bound for nonzero multivariate polynomials over finite fields, later used in polynomial identity testing (Schwartz--Zippel lemma).}
}

@article{AgrawalSSS16,
  author       = {Manindra Agrawal and
                  Chandan Saha and
                  Ramprasad Saptharishi and
                  Nitin Saxena},
  title        = {Jacobian Hits Circuits: Hitting Sets, Lower Bounds for Depth-D Occur-k
                  Formulas and Depth-3 Transcendence Degree-k Circuits},
  journal      = {{SIAM} J. Comput.},
  volume       = {45},
  number       = {4},
  pages        = {1533--1562},
  year         = {2016},
  url          = {https://doi.org/10.1137/130910725},
  doi          = {10.1137/130910725},
  bibsource    = {dblp computer science bibliography, https://dblp.org}
}

@article{KabanetsI04,
	Author = {Valentine Kabanets and Russell Impagliazzo},
	Bibsource = {dblp computer science bibliography, http://dblp.org},
	Doi = {10.1007/s00037-004-0182-6},
	Journal = {Computational Complexity},
	Number = {1-2},
	Pages = {1--46},
	Title = {Derandomizing Polynomial Identity Tests Means Proving Circuit Lower Bounds},
	Url = {http://dx.doi.org/10.1007/s00037-004-0182-6},
	Volume = {13},
	Year = {2004}}

@article{DvirSY09,
	author    = {Zeev Dvir and
	Amir Shpilka and
	Amir Yehudayoff},
	title     = {Hardness-Randomness Tradeoffs for Bounded Depth Arithmetic Circuits},
	journal   = {{SIAM} J. Comput.},
	volume    = {39},
	number    = {4},
	pages     = {1279--1293},
	year      = {2009},
	url       = {https://doi.org/10.1137/080735850},
	doi       = {10.1137/080735850},
	bibsource = {dblp computer science bibliography, https://dblp.org}
}

@article{GuoKSS22,
  author       = {Zeyu Guo and
                  Mrinal Kumar and
                  Ramprasad Saptharishi and
                  Noam Solomon},
  title        = {Derandomization from Algebraic Hardness},
  journal      = {{SIAM} J. Comput.},
  volume       = {51},
  number       = {2},
  pages        = {315--335},
  year         = {2022},
  url          = {https://doi.org/10.1137/20m1347395},
  doi          = {10.1137/20M1347395},
  bibsource    = {dblp computer science bibliography, https://dblp.org}
}

@article{KumarST23bootstrap,
  author       = {Mrinal Kumar and
                  Ramprasad Saptharishi and
                  Anamay Tengse},
  title        = {Near-Optimal Bootstrapping of Hitting Sets for Algebraic Models},
  journal      = {Theory Comput.},
  volume       = {19},
  pages        = {1--30},
  year         = {2023},
  url          = {https://doi.org/10.4086/toc.2023.v019a012},
  doi          = {10.4086/TOC.2023.V019A012},
  bibsource    = {dblp computer science bibliography, https://dblp.org}
}

@article{kumar2019hardness,
  title={Hardness-randomness tradeoffs for algebraic computation},
  author={Kumar, Mrinal and Saptharishi, Ramprasad},
  journal={Bulletin of EATCS},
  volume={3},
  number={129},
  year={2019}
}

@article{saxena2009progress,
  title={Progress on Polynomial Identity Testing.},
  author={Saxena, Nitin},
  journal={Bull. EATCS},
  volume={99},
  pages={49--79},
  year={2009}
}

@incollection{saxena2014progress,
  title={Progress on polynomial identity testing-II},
  author={Saxena, Nitin},
  booktitle={Perspectives in Computational Complexity: The Somenath Biswas Anniversary Volume},
  pages={131--146},
  year={2014},
  publisher={Springer}
}

@article{SY10,
	Author = {Shpilka, Amir and Yehudayoff, Amir},
	Journal = {Foundations and Trends{\textregistered} in Theoretical Computer Science},
	Number = {3--4},
	Pages = {207--388},
	Publisher = {Now Publishers, Inc.},
	Title = {Arithmetic circuits: A survey of recent results and open questions},
	Volume = {5},
	Year = {2010}}

@article{DuttaGhosh-survey,
author = {Pranjal Dutta and Sumanta Ghosh},
title = {SIGACT News Complexity Theory Column 121},
year = {2024},
issue_date = {June 2024},
publisher = {Association for Computing Machinery},
address = {New York, NY, USA},
volume = {55},
number = {2},
issn = {0163-5700},
url = {https://doi.org/10.1145/3674159.3674165},
doi = {10.1145/3674159.3674165},
journal = {SIGACT News},
month = jun,
pages = {53–88},
numpages = {36}
}

@article{beimel2000learning,
  title={Learning functions represented as multiplicity automata},
  author={Beimel, Amos and Bergadano, Francesco and Bshouty, Nader H and Kushilevitz, Eyal and Varricchio, Stefano},
  journal={Journal of the ACM (JACM)},
  volume={47},
  number={3},
  pages={506--530},
  year={2000},
  publisher={ACM New York, NY, USA}
}

@article{klivans2006learning,
  title={Learning restricted models of arithmetic circuits},
  author={Klivans, Adam and Shpilka, Amir},
  journal={Theory of computing},
  volume={2},
  number={1},
  pages={185--206},
  year={2006},
  publisher={Theory of Computing Exchange}
}

@article{BhargavaSV25,
  author       = {Vishwas Bhargava and
                  Shubhangi Saraf and
                  Ilya Volkovich},
  title        = {Reconstruction of Depth-4 Multilinear Circuits},
  journal      = {{ACM} Trans. Comput. Theory},
  volume       = {17},
  number       = {3},
  pages        = {17:1--17:23},
  year         = {2025},
  url          = {https://doi.org/10.1145/3726532},
  doi          = {10.1145/3726532},
  bibsource    = {dblp computer science bibliography, https://dblp.org}
}

@inproceedings{Sinha16,
  author       = {Gaurav Sinha},
  editor       = {Ran Raz},
  title        = {Reconstruction of Real Depth-3 Circuits with Top Fan-In 2},
  booktitle    = {31st Conference on Computational Complexity, {CCC} 2016, May 29 to
                  June 1, 2016, Tokyo, Japan},
  series       = {LIPIcs},
  volume       = {50},
  pages        = {31:1--31:53},
  publisher    = {Schloss Dagstuhl - Leibniz-Zentrum f{\"{u}}r Informatik},
  year         = {2016},
  url          = {https://doi.org/10.4230/LIPIcs.CCC.2016.31},
  doi          = {10.4230/LIPICS.CCC.2016.31},
  bibsource    = {dblp computer science bibliography, https://dblp.org}
}

@inproceedings{Sinha22,
  author       = {Gaurav Sinha},
  editor       = {Mark Braverman},
  title        = {Efficient Reconstruction of Depth Three Arithmetic Circuits with Top
                  Fan-In Two},
  booktitle    = {13th Innovations in Theoretical Computer Science Conference, {ITCS}
                  2022, January 31 - February 3, 2022, Berkeley, CA, {USA}},
  series       = {LIPIcs},
  volume       = {215},
  pages        = {118:1--118:33},
  publisher    = {Schloss Dagstuhl - Leibniz-Zentrum f{\"{u}}r Informatik},
  year         = {2022},
  url          = {https://doi.org/10.4230/LIPIcs.ITCS.2022.118},
  doi          = {10.4230/LIPICS.ITCS.2022.118},
  bibsource    = {dblp computer science bibliography, https://dblp.org}
}

@proceedings{DBLP:conf/eurosam/1979,
	Bibsource = {dblp computer science bibliography, http://dblp.org},
	Editor = {Edward W. Ng},
	Isbn = {3-540-09519-5},
	Publisher = {Springer},
	Series = {Lecture Notes in Computer Science},
	Title = {Symbolic and Algebraic Computation, {EUROSAM} '79, An International Symposiumon Symbolic and Algebraic Computation, Marseille, France, June 1979, Proceedings},
	Volume = {72},
	Year = {1979}}

@inproceedings{ForbesS12,
  author       = {Michael A. Forbes and
                  Amir Shpilka},
  editor       = {Howard J. Karloff and
                  Toniann Pitassi},
  title        = {On identity testing of tensors, low-rank recovery and compressed sensing},
  booktitle    = {Proceedings of the 44th Symposium on Theory of Computing Conference,
                  {STOC} 2012, New York, NY, USA, May 19 - 22, 2012},
  pages        = {163--172},
  publisher    = {{ACM}},
  year         = {2012},
  url          = {https://doi.org/10.1145/2213977.2213995},
  doi          = {10.1145/2213977.2213995},
  bibsource    = {dblp computer science bibliography, https://dblp.org}
}

@article{ShpilkaV14,
  author       = {Amir Shpilka and
                  Ilya Volkovich},
  title        = {On Reconstruction and Testing of Read-Once Formulas},
  journal      = {Theory Comput.},
  volume       = {10},
  pages        = {465--514},
  year         = {2014},
  url          = {https://doi.org/10.4086/toc.2014.v010a018},
  doi          = {10.4086/TOC.2014.V010A018},
  bibsource    = {dblp computer science bibliography, https://dblp.org}
}

@article{GuptaKQ14,
  author       = {Ankit Gupta and
                  Neeraj Kayal and
                  Youming Qiao},
  title        = {Random arithmetic formulas can be reconstructed efficiently},
  journal      = {Comput. Complex.},
  volume       = {23},
  number       = {2},
  pages        = {207--303},
  year         = {2014},
  url          = {https://doi.org/10.1007/s00037-014-0085-0},
  doi          = {10.1007/S00037-014-0085-0},
  bibsource    = {dblp computer science bibliography, https://dblp.org}
}

@inproceedings{GuptaKL12,
  author       = {Ankit Gupta and
                  Neeraj Kayal and
                  Satyanarayana V. Lokam},
  editor       = {Howard J. Karloff and
                  Toniann Pitassi},
  title        = {Reconstruction of depth-4 multilinear circuits with top fan-in 2},
  booktitle    = {Proceedings of the 44th Symposium on Theory of Computing Conference,
                  {STOC} 2012, New York, NY, USA, May 19 - 22, 2012},
  pages        = {625--642},
  publisher    = {{ACM}},
  year         = {2012},
  url          = {https://doi.org/10.1145/2213977.2214035},
  doi          = {10.1145/2213977.2214035},
  bibsource    = {dblp computer science bibliography, https://dblp.org}
}

@inproceedings{BhargavaGKS22,
  author       = {Vishwas Bhargava and
                  Ankit Garg and
                  Neeraj Kayal and
                  Chandan Saha},
  editor       = {Amit Chakrabarti and
                  Chaitanya Swamy},
  title        = {Learning Generalized Depth Three Arithmetic Circuits in the Non-Degenerate
                  Case},
  booktitle    = {Approximation, Randomization, and Combinatorial Optimization. Algorithms
                  and Techniques, {APPROX/RANDOM} 2022, September 19-21, 2022, University
                  of Illinois, Urbana-Champaign, {USA} (Virtual Conference)},
  series       = {LIPIcs},
  volume       = {245},
  pages        = {21:1--21:22},
  publisher    = {Schloss Dagstuhl - Leibniz-Zentrum f{\"{u}}r Informatik},
  year         = {2022},
  url          = {https://doi.org/10.4230/LIPIcs.APPROX/RANDOM.2022.21},
  doi          = {10.4230/LIPICS.APPROX/RANDOM.2022.21},
  bibsource    = {dblp computer science bibliography, https://dblp.org}
}

@inproceedings{ChandraGKMS24,
  author       = {Pritam Chandra and
                  Ankit Garg and
                  Neeraj Kayal and
                  Kunal Mittal and
                  Tanmay Sinha},
  editor       = {Venkatesan Guruswami},
  title        = {Learning Arithmetic Formulas in the Presence of Noise: {A} General
                  Framework and Applications to Unsupervised Learning},
  booktitle    = {15th Innovations in Theoretical Computer Science Conference, {ITCS}
                  2024, January 30 to February 2, 2024, Berkeley, CA, {USA}},
  series       = {LIPIcs},
  volume       = {287},
  pages        = {25:1--25:19},
  publisher    = {Schloss Dagstuhl - Leibniz-Zentrum f{\"{u}}r Informatik},
  year         = {2024},
  url          = {https://doi.org/10.4230/LIPIcs.ITCS.2024.25},
  doi          = {10.4230/LIPICS.ITCS.2024.25},
  bibsource    = {dblp computer science bibliography, https://dblp.org}
}

@inproceedings{AgrawalV08,
  author       = {Manindra Agrawal and
                  V. Vinay},
  title        = {Arithmetic Circuits: {A} Chasm at Depth Four},
  booktitle    = {49th Annual {IEEE} Symposium on Foundations of Computer Science, {FOCS}
                  2008, Philadelphia, PA, USA, October 25-28, 2008},
  pages        = {67--75},
  publisher    = {{IEEE} Computer Society},
  year         = {2008},
  url          = {https://doi.org/10.1109/FOCS.2008.32},
  doi          = {10.1109/FOCS.2008.32},
  bibsource    = {dblp computer science bibliography, https://dblp.org}
}

@article{Koiran12,
  author       = {Pascal Koiran},
  title        = {Arithmetic circuits: The chasm at depth four gets wider},
  journal      = {Theor. Comput. Sci.},
  volume       = {448},
  pages        = {56--65},
  year         = {2012},
  url          = {https://doi.org/10.1016/j.tcs.2012.03.041},
  doi          = {10.1016/J.TCS.2012.03.041},
  bibsource    = {dblp computer science bibliography, https://dblp.org}
}

@article{Tavenas15,
  author       = {S{\'{e}}bastien Tavenas},
  title        = {Improved bounds for reduction to depth 4 and depth 3},
  journal      = {Inf. Comput.},
  volume       = {240},
  pages        = {2--11},
  year         = {2015},
  url          = {https://doi.org/10.1016/j.ic.2014.09.004},
  doi          = {10.1016/J.IC.2014.09.004},
  bibsource    = {dblp computer science bibliography, https://dblp.org}
}

@article{GuptaKKS16chasm,
  author       = {Ankit Gupta and
                  Pritish Kamath and
                  Neeraj Kayal and
                  Ramprasad Saptharishi},
  title        = {Arithmetic Circuits: {A} Chasm at Depth 3},
  journal      = {{SIAM} J. Comput.},
  volume       = {45},
  number       = {3},
  pages        = {1064--1079},
  year         = {2016},
  url          = {https://doi.org/10.1137/140957123},
  doi          = {10.1137/140957123},
  bibsource    = {dblp computer science bibliography, https://dblp.org}
}

@article{Hastad90,
  author       = {Johan H{\aa}stad},
  title        = {Tensor Rank is NP-Complete},
  journal      = {J. Algorithms},
  volume       = {11},
  number       = {4},
  pages        = {644--654},
  year         = {1990},
  url          = {https://doi.org/10.1016/0196-6774(90)90014-6},
  doi          = {10.1016/0196-6774(90)90014-6},
  bibsource    = {dblp computer science bibliography, https://dblp.org}
}

@inproceedings{swernofsky2018tensor,
  title={Tensor rank is hard to approximate},
  author={Swernofsky, Joseph},
  booktitle={Approximation, Randomization, and Combinatorial Optimization. Algorithms and Techniques (APPROX/RANDOM 2018)},
  pages={26--1},
  year={2018},
  organization={Schloss Dagstuhl--Leibniz-Zentrum f{\"u}r Informatik}
}

@article{FortnowKlivans09,
  author  = {Lance Fortnow and Adam R.~Klivans},
  title   = {Efficient Learning Algorithms Yield Circuit Lower Bounds},
  journal = {Journal of Computer and System Sciences},
  year    = {2009},
  volume  = {75},
  number  = {1},
  pages   = {27--36},
  doi     = {10.1016/j.jcss.2008.09.010}
}

@article{KlivansSherstov09,
  author  = {Adam R.~Klivans and Alexander A.~Sherstov},
  title   = {Cryptographic Hardness for Learning Intersections of Halfspaces},
  journal = {Journal of Computer and System Sciences},
  year    = {2009},
  volume  = {75},
  number  = {1},
  pages   = {2--12},
  doi     = {10.1016/j.jcss.2008.07.018}
}

@article{SchaeferStefankovic18,
  author  = {Marcus Schaefer and Daniel {\v{S}}tefankovi{\v{c}}},
  title   = {The Complexity of Tensor Rank},
  journal = {Theory of Computing Systems},
  year    = {2018},
  volume  = {62},
  number  = {5},
  pages   = {1161--1174},
  doi     = {10.1007/s00224-017-9793-1}
}

@article{limaye2025superpolynomial,
  title={Superpolynomial lower bounds against low-depth algebraic circuits},
  author={Limaye, Nutan and Srinivasan, Srikanth and Tavenas, S{\'e}bastien},
  journal={Journal of the ACM},
  volume={72},
  number={4},
  pages={1--35},
  year={2025},
  publisher={ACM New York, NY}
}

@inproceedings{andrews2022ideals,
  title={Ideals, determinants, and straightening: proving and using lower bounds for polynomial ideals},
  author={Andrews, Robert and Forbes, Michael A},
  booktitle={Proceedings of the 54th Annual ACM SIGACT Symposium on Theory of Computing},
  pages={389--402},
  year={2022}
}

@article{GrigorievKS90,
  author       = {Dima Grigoriev and
                  Marek Karpinski and
                  Michael F. Singer},
  title        = {Fast Parallel Algorithms for Sparse Multivariate Polynomial Interpolation
                  over Finite Fields},
  journal      = {{SIAM} J. Comput.},
  volume       = {19},
  number       = {6},
  pages        = {1059--1063},
  year         = {1990},
  url          = {https://doi.org/10.1137/0219073},
  doi          = {10.1137/0219073},
  bibsource    = {dblp computer science bibliography, https://dblp.org}
}

@inproceedings{Ben-OrT88,
  author       = {Michael Ben{-}Or and
                  Prasoon Tiwari},
  editor       = {Janos Simon},
  title        = {A Deterministic Algorithm for Sparse Multivariate Polynomial Interpolation
                  (Extended Abstract)},
  booktitle    = {Proceedings of the 20th Annual {ACM} Symposium on Theory of Computing,
                  May 2-4, 1988, Chicago, Illinois, {USA}},
  pages        = {301--309},
  publisher    = {{ACM}},
  year         = {1988},
  url          = {https://doi.org/10.1145/62212.62241},
  doi          = {10.1145/62212.62241},
  bibsource    = {dblp computer science bibliography, https://dblp.org}
}

@article{DvirS07,
  author       = {Zeev Dvir and
                  Amir Shpilka},
  title        = {Locally Decodable Codes with Two Queries and Polynomial Identity Testing
                  for Depth 3 Circuits},
  journal      = {{SIAM} J. Comput.},
  volume       = {36},
  number       = {5},
  pages        = {1404--1434},
  year         = {2007},
  url          = {https://doi.org/10.1137/05063605X},
  doi          = {10.1137/05063605X},
  bibsource    = {dblp computer science bibliography, https://dblp.org}
}

@inproceedings{Forbes15,
  author       = {Michael A. Forbes},
  editor       = {Venkatesan Guruswami},
  title        = {Deterministic Divisibility Testing via Shifted Partial Derivatives},
  booktitle    = {{IEEE} 56th Annual Symposium on Foundations of Computer Science, {FOCS}
                  2015, Berkeley, CA, USA, 17-20 October, 2015},
  pages        = {451--465},
  publisher    = {{IEEE} Computer Society},
  year         = {2015},
  url          = {https://doi.org/10.1109/FOCS.2015.35},
  doi          = {10.1109/FOCS.2015.35},
  bibsource    = {dblp computer science bibliography, https://dblp.org}
}

@inproceedings{GuoG20,
  author       = {Zeyu Guo and
                  Rohit Gurjar},
  editor       = {Jaroslaw Byrka and
                  Raghu Meka},
  title        = {Improved Explicit Hitting-Sets for ROABPs},
  booktitle    = {Approximation, Randomization, and Combinatorial Optimization. Algorithms
                  and Techniques, {APPROX/RANDOM} 2020, August 17-19, 2020, Virtual
                  Conference},
  series       = {LIPIcs},
  volume       = {176},
  pages        = {4:1--4:16},
  publisher    = {Schloss Dagstuhl - Leibniz-Zentrum f{\"{u}}r Informatik},
  year         = {2020},
  url          = {https://doi.org/10.4230/LIPIcs.APPROX/RANDOM.2020.4},
  doi          = {10.4230/LIPICS.APPROX/RANDOM.2020.4},
  bibsource    = {dblp computer science bibliography, https://dblp.org}
}

@article{gurjar2016identity,
  title={Identity testing for constant-width, and any-order, read-once oblivious arithmetic branching programs},
  author={Gurjar, Rohit and Korwar, Arpita and Saxena, Nitin},
  journal={arXiv preprint arXiv:1601.08031},
  year={2016}
}

@inproceedings{peleg2021polynomial,
  title={Polynomial time deterministic identity testing algorithm for $\Sigma$ [3] $\Pi$$\Sigma$$\Pi$ [2] circuits via Edelstein--Kelly type theorem for quadratic polynomials},
  author={Peleg, Shir and Shpilka, Amir},
  booktitle={Proceedings of the 53rd Annual ACM SIGACT Symposium on Theory of Computing},
  pages={259--271},
  year={2021}
}

@article{GargOS25a,
  author       = {Abhibhav Garg and
                  Rafael Mendes de Oliveira and
                  Akash Kumar Sengupta},
  title        = {Rank Bounds and {PIT} for {\textdollar}{\textbackslash}Sigma{\^{}}3
                  {\textbackslash}Pi {\textbackslash}Sigma {\textbackslash}Pi{\^{}}d{\textdollar}
                  circuits via a non-linear Edelstein-Kelly theorem},
  journal      = {Electron. Colloquium Comput. Complex.},
  volume       = {{TR25-051}},
  year         = {2025},
  url          = {https://eccc.weizmann.ac.il/report/2025/051},
  eprinttype    = {ECCC},
  eprint       = {TR25-051},
  bibsource    = {dblp computer science bibliography, https://dblp.org}
}

@inproceedings{DuttaD021,
  author       = {Pranjal Dutta and
                  Prateek Dwivedi and
                  Nitin Saxena},
  editor       = {Valentine Kabanets},
  title        = {Deterministic Identity Testing Paradigms for Bounded Top-Fanin Depth-4
                  Circuits},
  booktitle    = {36th Computational Complexity Conference, {CCC} 2021, July 20-23,
                  2021, Toronto, Ontario, Canada (Virtual Conference)},
  series       = {LIPIcs},
  volume       = {200},
  pages        = {11:1--11:27},
  publisher    = {Schloss Dagstuhl - Leibniz-Zentrum f{\"{u}}r Informatik},
  year         = {2021},
  url          = {https://doi.org/10.4230/LIPIcs.CCC.2021.11},
  doi          = {10.4230/LIPICS.CCC.2021.11},
  bibsource    = {dblp computer science bibliography, https://dblp.org}
}

@inproceedings{saraf2025reconstruction,
  title={Reconstruction of depth 3 arithmetic circuits with top fan-in 3},
  author={Saraf, Shubhangi and Shringi, Devansh},
  booktitle={40th Computational Complexity Conference (CCC 2025)},
  pages={21--1},
  year={2025},
  organization={Schloss Dagstuhl--Leibniz-Zentrum f{\"u}r Informatik}
}

@inproceedings{GeHK15,
  author       = {Rong Ge and
                  Qingqing Huang and
                  Sham M. Kakade},
  editor       = {Rocco A. Servedio and
                  Ronitt Rubinfeld},
  title        = {Learning Mixtures of Gaussians in High Dimensions},
  booktitle    = {Proceedings of the Forty-Seventh Annual {ACM} on Symposium on Theory
                  of Computing, {STOC} 2015, Portland, OR, USA, June 14-17, 2015},
  pages        = {761--770},
  publisher    = {{ACM}},
  year         = {2015},
  url          = {https://doi.org/10.1145/2746539.2746616},
  doi          = {10.1145/2746539.2746616},
  bibsource    = {dblp computer science bibliography, https://dblp.org}
}

@article{anandkumar2014tensor,
  title={Tensor decompositions for learning latent variable models.},
  author={Anandkumar, Animashree and Ge, Rong and Hsu, Daniel J and Kakade, Sham M and Telgarsky, Matus and others},
  journal={J. Mach. Learn. Res.},
  volume={15},
  number={1},
  pages={2773--2832},
  year={2014}
}

@inproceedings{goyal2014fourier,
  title={Fourier PCA and robust tensor decomposition},
  author={Goyal, Navin and Vempala, Santosh and Xiao, Ying},
  booktitle={Proceedings of the forty-sixth annual ACM symposium on Theory of computing},
  pages={584--593},
  year={2014}
}

@inproceedings{ma2016polynomial,
  title={Polynomial-time tensor decompositions with sum-of-squares},
  author={Ma, Tengyu and Shi, Jonathan and Steurer, David},
  booktitle={2016 IEEE 57th Annual Symposium on Foundations of Computer Science (FOCS)},
  pages={438--446},
  year={2016},
  organization={IEEE}
}

@inproceedings{hopkins2015tensor,
  title={Tensor principal component analysis via sum-of-square proofs},
  author={Hopkins, Samuel B and Shi, Jonathan and Steurer, David},
  booktitle={Conference on Learning Theory},
  pages={956--1006},
  year={2015},
  organization={PMLR}
}

@inproceedings{hopkins2016fast,
  title={Fast spectral algorithms from sum-of-squares proofs: tensor decomposition and planted sparse vectors},
  author={Hopkins, Samuel B and Schramm, Tselil and Shi, Jonathan and Steurer, David},
  booktitle={Proceedings of the forty-eighth annual ACM symposium on Theory of Computing},
  pages={178--191},
  year={2016}
}

@article{montanari2014statistical,
  title={A statistical model for tensor PCA},
  author={Montanari, Andrea and Richard, Emile},
  journal={Advances in neural information processing systems},
  volume={27},
  year={2014}
}

@article{kothari2025smooth,
  title={Smooth Trade-off for Tensor PCA via Sharp Bounds for Kikuchi Matrices},
  author={Kothari, Pravesh K and Xu, Jeff},
  journal={arXiv preprint arXiv:2510.03061},
  year={2025}
}

@inproceedings{bafna2022polynomial,
  title={Polynomial-time power-sum decomposition of polynomials},
  author={Bafna, Mitali and Hsieh, Jun-Ting and Kothari, Pravesh K and Xu, Jeff},
  booktitle={2022 IEEE 63rd Annual Symposium on Foundations of Computer Science (FOCS)},
  pages={956--967},
  year={2022},
  organization={IEEE}
}

@inproceedings{bhaskara2024new,
  title={New tools for smoothed analysis: Least singular value bounds for random matrices with dependent entries},
  author={Bhaskara, Aditya and Evert, Eric and Srinivas, Vaidehi and Vijayaraghavan, Aravindan},
  booktitle={Proceedings of the 56th Annual ACM Symposium on Theory of Computing},
  pages={375--386},
  year={2024}
}

@inproceedings{anderson2014more,
  title={The more, the merrier: the blessing of dimensionality for learning large Gaussian mixtures},
  author={Anderson, Joseph and Belkin, Mikhail and Goyal, Navin and Rademacher, Luis and Voss, James},
  booktitle={Conference on Learning Theory},
  pages={1135--1164},
  year={2014},
  organization={PMLR}
}

@article{dasgupta2007probabilistic,
  title={A Probabilistic Analysis of EM for Mixtures of Separated, Spherical Gaussians.},
  author={Dasgupta, Sanjoy and Schulman, Leonard},
  journal={Journal of Machine Learning Research},
  volume={8},
  number={2},
  year={2007}
}

@article{di2023multidimensional,
  title={The multidimensional truncated moment problem: Gaussian mixture reconstruction from derivatives of moments},
  author={di Dio, Philipp J},
  journal={Journal of Mathematical Analysis and Applications},
  volume={517},
  number={1},
  pages={126592},
  year={2023},
  publisher={Elsevier}
}

@inproceedings{liu2021settling,
  title={Settling the robust learnability of mixtures of gaussians},
  author={Liu, Allen and Moitra, Ankur},
  booktitle={Proceedings of the 53rd Annual ACM SIGACT Symposium on Theory of Computing},
  pages={518--531},
  year={2021}
}

@inproceedings{regev2017learning,
  title={On learning mixtures of well-separated gaussians},
  author={Regev, Oded and Vijayaraghavan, Aravindan},
  booktitle={2017 IEEE 58th Annual Symposium on Foundations of Computer Science (FOCS)},
  pages={85--96},
  year={2017},
  organization={IEEE}
}

@inproceedings{sanjeev2001learning,
  title={Learning mixtures of arbitrary Gaussians},
  author={Sanjeev, Arora and Kannan, Ravi},
  booktitle={Proceedings of the thirty-third annual ACM symposium on Theory of computing},
  pages={247--257},
  year={2001}
}

@article{sylvester1851lx,
  title={Lx. on a remarkable discovery in the theory of canonical forms and of hyperdeterminants},
  author={Sylvester, James Joseph},
  journal={The London, Edinburgh, and Dublin Philosophical Magazine and Journal of Science},
  volume={2},
  number={12},
  pages={391--410},
  year={1851},
  publisher={Taylor \& Francis}
}

@inproceedings{HS80,
  title={Testing polynomials which are easy to compute},
  author={Heintz, Joos and Schnorr, Claus-Peter},
  booktitle={Proceedings of the twelfth annual ACM Symposium on Theory of Computing},
  pages={262--272},
  year={1980}
}

@article{Chou0S19,
  author       = {Chi{-}Ning Chou and
                  Mrinal Kumar and
                  Noam Solomon},
  title        = {Closure Results for Polynomial Factorization},
  journal      = {Theory Comput.},
  volume       = {15},
  pages        = {1--34},
  year         = {2019},
  url          = {https://doi.org/10.4086/toc.2019.v015a013},
  doi          = {10.4086/TOC.2019.V015A013},
  bibsource    = {dblp computer science bibliography, https://dblp.org}
}

@article{GuoWang25,
  author       = {Zeyu Guo and
                  Siki Wang},
  title        = {Deterministic Depth-4 {PIT} and Normalization},
  journal      = {CoRR},
  volume       = {abs/2504.15143},
  year         = {2025},
  url          = {https://doi.org/10.48550/arXiv.2504.15143},
  doi          = {10.48550/ARXIV.2504.15143},
  eprinttype    = {arXiv},
  eprint       = {2504.15143},
  bibsource    = {dblp computer science bibliography, https://dblp.org}
}

@article{SarafSV25,
  author       = {Shubhangi Saraf and
                  Devansh Shringi and
                  Narmada Varadarajan},
  title        = {Reconstruction of Depth-3 Arithmetic Circuits with Constant Top Fan-in},
  journal      = {Electron. Colloquium Comput. Complex.},
  volume       = {{TR25-222}},
  year         = {2025},
  url          = {https://eccc.weizmann.ac.il/report/2025/222},
  eprinttype    = {ECCC},
  eprint       = {TR25-222},
  bibsource    = {dblp computer science bibliography, https://dblp.org}
}

@inproceedings{ChillaraGS23,
  author       = {Suryajith Chillara and
                  Coral Grichener and
                  Amir Shpilka},
  editor       = {Petra Berenbrink and
                  Patricia Bouyer and
                  Anuj Dawar and
                  Mamadou Moustapha Kant{\'{e}}},
  title        = {On Hardness of Testing Equivalence to Sparse Polynomials Under Shifts},
  booktitle    = {40th International Symposium on Theoretical Aspects of Computer Science,
                  {STACS} 2023, Hamburg, Germany, March 7-9, 2023},
  series       = {LIPIcs},
  volume       = {254},
  pages        = {22:1--22:20},
  publisher    = {Schloss Dagstuhl - Leibniz-Zentrum f{\"{u}}r Informatik},
  year         = {2023},
  url          = {https://doi.org/10.4230/LIPIcs.STACS.2023.22},
  doi          = {10.4230/LIPICS.STACS.2023.22},
  bibsource    = {dblp computer science bibliography, https://dblp.org}
}

\appendix
\section{Linear Independence and Wronskian}\label{app:wronskian}

We prove \autoref{thm:wronskian-lin-ind}. We recall its statement.

\thmwronskilinind*

\begin{proof}
The 'if' direction is immediate.  
For the converse, assume that $g_1, \dots, g_n$ are linearly independent.  
Since the determinant is invariant under elementary column operations, we may assume, without loss of generality, that $g_1, \dots, g_n$ have distinct degrees.

Let their degrees be $d_1 > d_2 > \cdots > d_n \ge 0$, and write  
$g_j = c_j x^{d_j} + \text{(lower terms)}$ with $c_j \ne 0$.  
By multilinearity of the determinant, the leading term of the Wronskian equals the Wronskian of the leading monomials, provided its coefficient is nonzero:
\[
W(\mathrm{LM}(g_1), \dots, \mathrm{LM}(g_n))
  = W(c_1 x^{d_1}, \dots, c_n x^{d_n}).
\]

The determinant of the matrix of derivatives of $c_j x^{d_j}$ is
\[
\det\!\big( c_j (d_j)_{i}\, x^{d_j - i} \big)_{\substack{i\in [[n-1]] \\ j\in [n]}},
\]
where $(d_j)_i = d_j(d_j-1)\cdots(d_j - i + 1)$ is the falling factorial (defined to be $1$ for $i=0$).
Factoring out $c_j$ from each column and collecting powers of $x$ yields
\[
W(c_1 x^{d_1}, \dots, c_n x^{d_n})
  = \left(\Pi_{j=1}^n c_j\right)
    \left(\Pi_{1 \le i < j \le n} (d_j - d_i)\right)
    \left(\Pi_{j=1}^n (d_j)_{n-1}\right)
    x^{\sum_j d_j - \binom{n}{2}}.
\]

Because $p=0$ or $p > d$, all integers $d_j$ and differences $(d_j - d_i)$ are nonzero in $\F$, and each falling factorial $(d_j)_{n-1}$ is nonzero.  
Since each $c_j \ne 0$, the product above is nonzero, so $W(g_1, \dots, g_n)$ is not identically zero.
\end{proof}

\end{document}